%% file: main.tex
\documentclass[conference]{IEEEtran}
\usepackage{times}

\usepackage[numbers]{natbib}
\usepackage{multicol}
\usepackage[bookmarks=true]{hyperref}

\input{packages}
\input{commands}

\begin{document}

\title{Safe Importance Sampling in Model Predictive Path Integral Control}


\author{\authorblockN{Manan Gandhi}
\authorblockA{School of Aerospace Engineering\\
Georgia Institute of Technology\\
Atlanta, Georgia 30332--0250\\
Email: mgandhi@gatech.edu}
\and
\authorblockN{Hassan Almubarak}
\authorblockA{School of Electrical \\and Computer Engineering\\
Georgia Institute of Technology\\
Atlanta, Georgia 30332--0250\\
Email: halmubarak@gatech.edu}
\and
\authorblockN{Evangelos Theodorou}
\authorblockA{School of Aerospace Engineering\\
Georgia Institute of Technology\\
Atlanta, Georgia 30332--0250\\
Email: evangelos.theodorou@gatech.edu}}


%

\maketitle

\begin{abstract}
We introduce the notion of importance sampling under embedded barrier state control, titled \ac{SC-MPPI}. For robotic systems operating in an environment with multiple constraints, hard constraints are often encoded utilizing penalty functions when performing optimization. Alternative schemes utilizing optimization-based techniques, such as Control Barrier Functions, can be used as a safety filter to ensure the system does not violate the given hard constraints. In contrast, this work leverages the principle of a safety filter but applies it during forward sampling for Model Predictive Path Integral Control. The resulting set of forward samples can remain safe within the domain of the safety controller, increasing sample efficiency and allowing for improved exploration of the state space. We derive this controller through information theoretic principles analogous to Information Theoretic MPPI. We empirically demonstrate both superior sample efficiency, exploration, and system performance of \ac{SC-MPPI} when compared to \ac{MPPI} and \ac{DDP} optimizing the barrier state.
\end{abstract}

\IEEEpeerreviewmaketitle

\section{Introduction}
Safety-critical control is a fundamental problem in dynamical systems with many problems in robotics, healthcare, and aviation requiring safe operation. In the field of terrestrial and aerial agility, sampling-based control \cite{williams2017information, foehn2022alphapilot} has been utilized to achieve high performing, aggressive control structures, however hard constraints for these systems is often implemented in terms of a safety filter or as penalty functions in the optimization. In this work we will present \ac{SC-MPPI}, an algorithm to embed safety into the sampling phase of \ac{MPPI}. 

\begin{figure}[t] 
    \centering
  \subfloat[]{\includegraphics[trim={5cm 1cm 5cm 3cm},clip,width=0.45\linewidth]{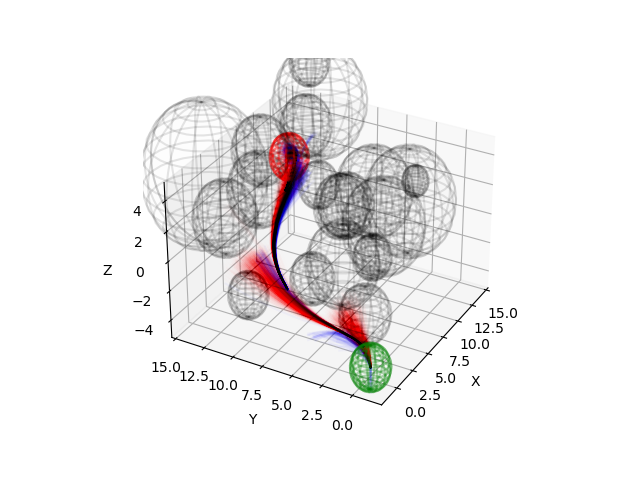}}
  \subfloat[]{\includegraphics[trim={5cm 1cm 5cm 3cm},clip,width=0.45\linewidth]{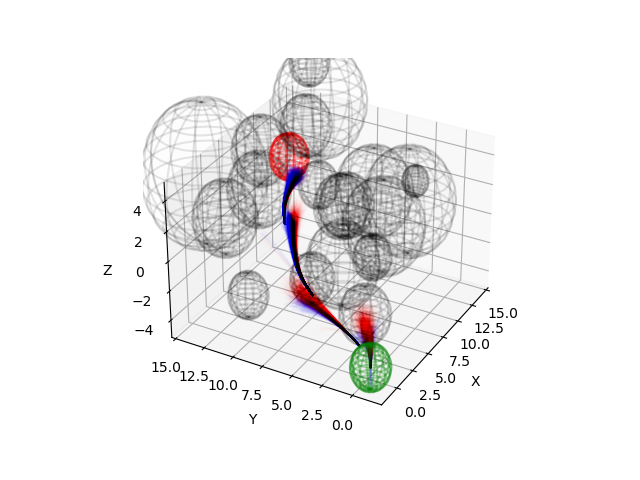}}
  \caption{Quadrotor in a dense obstacle field with visualization of the sampling. Safe samples are shown in blue, and unsafe samples are shown in red. The green sphere is the initial position, the red sphere is the desired final position. The left figure is \ac{SC-MPPI} shown with obstacles. The right figure is \ac{MPPI}. Note that \ac{MPPI} trajectories are closer to the nominal trajectory, while \ac{SC-MPPI} trajectories project further forward while moving away from the obstacles.}
  \label{fig: quadrotor sampling} 
\end{figure}

In the context of safe control, we review some candidate techniques for maintaining the safety of a known system model through feedback. This feedback forms the lynch pin of \ac{SC-MPPI}, as the feedback mechanism in question will be applied on each sample individually (see \Cref{fig: quadrotor sampling}). Potential field methods utilize attractive forces for a given goal state and repulsive forces generated by obstacles. The difficulty in potential field methods lies in optimizing the ratio between the two forces, additionally, the method itself is subject to becoming trapped before narrow passages, as demonstrated in \citet{koren1991potential}.  \citet{singletary2020comparative} demonstrate that potential functions are a subset of Control Barrier Functions, whereas Control Barrier Functions can be utilized in a more general fashion as a safety filter. \ac{CBFs} have been widely successful in tasks such as bipedal walking \cite{agrawal2017discrete, ames2019control} and automatic lane keeping \cite{ames2016control}, however for complex constraints, the control is the result of an optimization scheme that is computationally expensive to run in real-time thousands of samples. Additionally, a particular structure is required to handle systems with high relative degree \cite{xiao2019control} and in the context of discrete \ac{CBFs}, even linear systems with linear constraints can result in a non-convex optimization problem \cite{agrawal2017discrete}. Even in the realm of sampling-based control, \ac{CBFs} can offer a useful method to improve sample efficiency, but is held back again due to computational complexity \cite{tao2021control}. Embedded barrier state (BaS) methods developed by \citet{almubarak2022safety} append the state space with the state of a barrier function, then utilize a Lyapunov stability criterion, satisfied by an optimal controller, to ensure safety through stabilization of the appended model. Through similar arguments, \citet{almubarak2022safeddp} propose using discrete barrier states in discrete time trajectory optimization settings and show that the safety embedding technique exhibits great performance improvements over penalty methods which, implicitly, result in similar cost functions. This method is applicable to general nonlinear systems, and can combine a variety of nonlinear constraints at the expense of sensitivity to gradients of the barrier state dynamics, time discretization, and increasing the model's dimension. A great advantage of designing a feedback controller for the embedded system is the fact that the controller now is a function of the barrier. In the proposed work, we utilize discrete embedded barrier state feedback control as a means of safety in sampling, due to relaxation in problem assumptions, computational complexity, and ability to combine a large number of disjoint, unsafe regions.

Safety in sampling based model predictive control has received recent attention from \citet{martin2019trustregions} and \citet{zeng2021safety}, due to the high performance capabilities of sampling-based MPC \cite{williams2017information}. Additionally, sampling based trajectory optimization shares a close connection with safe model-based reinforcement learning, which can be loosely categorized into two categories: safety within the optimization metric, and safety in sample exploration by \citet{garcia2015comprehensive}. Safety in importance sampling falls into the latter category, in the same vein as the work by \citet{berkenkamp2017safe} which utilizes Lyapunov-based stability metrics to explore the policy space, the work by \citet{thananjeyan2021abc}, who utilize sampling to iteratively improve the policy for a nonlinear system. Additionally, there is the work in covariance steering control in combination with \ac{MPPI} by \citet{yin2022trajectory}. In this work, the authors formulate a convex optimization problem to solve for a feedback control which satisfies an upper bound on the covariance of the terminal states. While in CC-MPPI, a feedback controller is utilized to guide trajectory samples away from high cost areas, the mechanism in the proposed work relies of feedback from the barrier state, enabling safety without having a soft cost on obstacles. CC-MPPI is also computationally intensive, with the controller given in \cite{yin2022trajectory} running at 13 Hz, due to the complexity of solving the covariance steering problem around the given reference trajectory. In our work, the safety embedded controller is utilized to \textit{both modify the reference trajectory} (if needed), \textit{and compute feedback gains} on a barrier state in order to establish safety, all while maintaining \textit{optimization times below 10 ms}.

There is also the use of constrained covariance steering in Tube-Based MPPI by \citet{9867864}, where the authors apply a constrained covariance steering controller as a probabilistic safety filter on top of MPPI. This is again fundamentally different from our work, since safety is not applied during the sampling phase. In the attempt to unify the safety critical control and sampling-based MPC, we develop a new algorithm based on information theoretic MPC that encodes knowledge of the safe controller into the forward sampling. We show \textbf{three key contributions} in this work:

\begin{enumerate}
    \item Derivation of a new control scheme for embedding safety into MPPI.
    \item Empirical results of both computational efficiency (real-time performance), and sample efficiency (\% of collision free samples).
    \item Superior performance of the proposed algorithm versus existing methods in a navigation task in a cluttered environment, with respect to vehicle speed, task completion and final position error.
\end{enumerate}

\section{Mathematical Background}
For the background of this work, we will first review the problem at hand, and then describe the fundamentals of our choice of safety controller: Differential Dynamic Programming with \textit{Embedded Barrier States} by \citet{almubarak2022safeddp}. The decision to utilize Embedded Barrier States for safety was due to the flexibility and computational ease of the framework. Alternative safety schemes can be utilized under this same proposed framework, typically with a significant increase in computational complexity. Next we will review the fundamentals of Information Theoretic Path Integral Control, then delve into the proposed algorithm.

Consider the discrete, nonlinear dynamical system:
\begin{align}
    \vx_{k+1} &= F(k, \vx_k, \vu_k ), \label{eq: dynamics}
\end{align}
where, at a time step $k \in \real^+$, $\vx \in \domain \subset \real^n$, $\vu \in \real^m$, and $F: (\real^+ \times \real^n \times \real^m) \rightarrow \real^n$. Within the domain $\domain$, we can define a \textit{safe set} $\safeset \subset \domain$. The goal of the model predictive controller will be to achieve an objective in finite time while keeping the state trajectory, $\xtraj=\{ \vx_0, \vx_1, ..., \vx_T \}$ inside the safe set $\safeset$. Specifically, the model predictive control problem considers minimizing the cost functional 
\begin{align}
    J(\xtraj, \utraj) &= \phi(\vx_T) + \sum_{k=0}^{T-1} \big( q(\vx_k, k) + \lambda \vu^{\rT} \Sigma^{-1} \vu \big), \label{eq: generic cost}
\end{align}
subject to safety state constraints and control limits with $\utraj=\{ \vu_0, \vu_1, ..., \vu_{T-1} \}$. Here $\phi$ is a terminal cost, and $q$ is a nonlinear, potentially time varying state cost with a quadratic control cost penalty. $\lambda$ is known as the inverse temperature and $\Sigma \in \real^{m \times m}$ is a positive definite control penalty matrix. For the safety state constraints, we consider the superlevel set $\safeset \subset \real^n$ defined by a continuously differentiable real valued function $h: \domain\subset \mathbb{R}^n \rightarrow \mathbb{R}$ such that the set $\safeset$, its interior $\safeset^{\circ}$, and boundary $\partial \safeset$ are defined respectively as
\begin{align}  \label{eq: control barrier safe set} 
\begin{split}
    \safeset &= \{\vx  \in \domain: h(\vx) \geq 0\},\\
    \safeset^{\circ} &= \{\vx  \in \domain: h(\vx) > 0\},\\
    \partial \safeset &= \{\vx  \in \domain: h(\vx) = 0\}.
\end{split}
\end{align}
To ensure safety, the safe set $\safeset$ needs to be rendered controlled forward invariant, i.e. the system's safety critical states never leave the set $\safeset$. The notion of \textit{forward invariance} of the safe set $\safeset$ with respect to the safety-critical dynamical system \eqref{eq: dynamics} can be formally defined as follows \cite{blanchini1999set}:
\begin{definition} The set $\safeset \subset \real^n$ is said to be \textit{controlled forward invariant} for the dynamical system $\vx_{k+1} = F(k, \vx_k, \vu_k )$, if, for all $x_0 \in \safeset$, there exists a feedback control $\vu_k = K_{\text{safe}}(k, \vx_k)$, such that $\vx_{k+1} = F(k, \vx_k, K_{\text{safe}}(k, \vx_k)) \in \safeset$, for all $k \in [t_0, T]$.
\end{definition}

\subsection{Embedded Barrier States} \label{subsec: embedded barrier states}
Consider the undisturbed system in \eqref{eq: dynamics} and the safety constraint \eqref{eq: control barrier safe set}. A smooth scalar valued function $B:\safeset^{\circ} \rightarrow \mathbb{R}$ defined over $h$, $B\big(h(\vx_k)\big)$, is a barrier function if $B \rightarrow \infty $ as $\vx_k \rightarrow \partial\safeset$. In control barrier functions \cite{prajna2004safety,wieland2007constructive,romdlony2016stabilization,ames2016control}, to ensure boundedness of the barrier function, which implies satisfaction of the safety condition $h(\vx_x)>0 $, the barrier's rate of change is required to be decreasing (or not exceed a certain value associated with the barrier function as proposed in \cite{ames2016control}). The authors in \cite{almubarak2022safety} proposed \textit{barrier states} (BaS) to transform the safety objective into a performance objective by embedding the state of the barrier into the model of the safety-critical system. In essence, the safe control synthesis is coupled with the control design of other performance objectives. The idea is that the barrier function's rate of change is controlled along the system's state to ensure its boundedness in lieu of enforcing this rate through an inequality hard constraint as in CBFs. In \cite{almubarak2022safety}, the barrier state embedded model is asymptotically stabilized, which implies safety due to boundedness of the barrier state (see \cite{almubarak2022safety} Theorem 3). Next, the authors proposed discrete barrier states (DBaS) with differential dynamic programming to perform safe trajectory optimization \cite{almubarak2022safeddp}, which was shown to greatly simplified the problem formulation. 

Defining the barrier over the safety condition as $\beta_k:= B\big(h(\vx_k)\big)$, the DBaS dynamics are defined as
\begin{align} \label{eq: discrete barrier state}
    F^{\beta} (\vx_k, \beta_k, \vu_k) &:= \beta_{k+1} \\
                                      &= B \circ h \circ F(k, \vx_k, \vu_k ) \nonumber \\
                                      &- \gamma (\beta_k - B \circ h \circ \vx_k). \nonumber
\end{align}
The additional term $\gamma (\beta_k - B \circ h \circ \vx_k)$, parameterized by $ |\gamma| \leq 1$, is the DBaS pole for the linearized system. This then guarantees a non-vanishing gradient of the barrier, ensuring a non-zero feedback. Notice that this term is essentially zero by definition of the barrier (see the detailed proof of the continuous time case in \cite{almubarak2022safety}). 
Details and a numerical example of the barrier states feedback, as well as how it appears in differential dynamic programming (DDP) are provided in Appendix \ref{App subsec: DBaS in trajectory optimization}. This algorithm is used to generate feedback gains for importance sampling in Section \ref{sec: Safe Information Theoretic Model Predictive Control }. For multiple constraints, multiple barrier functions can be added to form a single barrier \cite{almubarak2022safeddp} or multiple barrier states. Then, the barrier state vector $\beta \in \mathcal{B} \subset \mathbb{R}^{n_\beta}$, where $n_\beta$ is the dimensionality of the barrier state vector, is appended to the dynamical model resulting in the safety embedded system:
\begin{equation} \label{eq: safety embedded model}
    \bar{\vx}_{k+1} = \bar{F}(k, \bar{\vx}_k, \vu_k ), 
\end{equation}
where $\bar{F} = \begin{bmatrix}F(k,\vx_k, \vu_k), & F^{\beta}\end{bmatrix}^\text{T}$ and $\bar{\vx} = \begin{bmatrix} \vx, & \beta\end{bmatrix}^\text{T}$.

One of the benefits of a safety embedded model is the direct transmission of safety constraint information to the optimal controller (see Appendix \ref{App subsec: DBaS in trajectory optimization}). This prevents two separate algorithms from \textit{fighting} one another for control bandwidth, i.e. a controller attempting to maximize performance and a safety filter attempting to maximize safety. This comes at a cost of the user having to specify the weighting between task performance and safety, similar to barrier methods in optimization. For the model predictive control problem in this work, the following proposition \cite{almubarak2022safeddp} depicts the safety guarantees provided by the embedded barrier state method.

\begin{proposition}[\cite{almubarak2022safeddp}] \label{prop:safety}
Under the control sequence $\utraj$, the safe set $\safeset$ is controlled forward invariant if and only if $\beta(\vx(0)) < \infty \Rightarrow \beta_k <\infty \ \forall k \in [1, T]$.
\end{proposition}

For the optimal control problem considered in this work \Cref{eq: dynamics}, \Cref{eq: generic cost}, with the constraints in \Cref{eq: control barrier safe set}, the embedded barrier state paradigm transforms the problem to the following:
\begin{align} \begin{split}
    \min_{\utraj}  \sum_{k=0}^{T-1} & \big( q(\bar{\vx}_k, k) + \vu^{\rT} \Sigma^{-1} \vu \big) + \phi(\bar{\vx}_T) \\
    \text{subject to } & \bar{\vx}_{k+1} = \bar{F}(k, \bar{\vx}_k, \vu_k )
\end{split} \label{eq: bas optimal control problem} \end{align}
This transformation was first proposed in \cite{almubarak2022safeddp} and used in \cite{cho2023MPC-DBaS_DDP} with MPC for safe driving of autonomous vehicles. Note that in this work, we use embedded barrier state DDP to generate safe reference trajectories and safe feedback gains along the reference trajectories. The safe feedback gains are applied on the barrier state $\beta$ to guide samples for \ac{MPPI} away from obstacles.

\subsection{Model Predictive Path Integral Control}
We briefly review Free Energy, Relative Entropy, and the connection to model predictive control. Additional details can be found in \cite{williams2017information}. First, we define free energy with the following expression:
\begin{align}
\mathcal{F}(S, \mathbb{P}, \vx_0, \lambda) &= -\lambda \log \Bigg[ \mathbb{E}_{\mathbb{P}}(\exp(-\frac{1}{\lambda} S(V)) \Bigg], \label{eq: free energy} \\
    S(V, \vx_0) &= \phi(\vx_T) + \sum_{t = 0}^{T-1} q(\vx_t) \label{eq:pathcost}.
\end{align}
 $\lambda$ is the inverse temperature. $S$ is the cost-to-go function, which takes in initial condition $\vx_0$ and set of random variables that generate a sequence of controls $V = \{ \vv_0, \vv_1, ..., \vv_{T-1} \}$, which in turn generate a trajectory $\xtraj$ evaluated with the terminal and state cost functions $\phi$ and $q$ respectively. The probability measure $\Pb$ is utilized to sample the controls of the system when computing the free energy.

We can upper bound the free energy using Jenson's Inequality,
\begin{align}
    \text{KL}(\mathbb{Q} || \mathbb{P}) &= \mathbb{E}_{\mathbb{Q}}\Big[ \log (\frac{\text{d}\mathbb{Q}}{\text{d}\mathbb{P}}) \Big], \\
    \mathcal{F}(S, \mathbb{P}, \vx_0, \lambda) &\leq \mathbb{E}_{\mathbb{Q}}\Big[ S(V) \Big] + \lambda \text{KL}(\mathbb{Q} || \mathbb{P}). \label{eq: free energy relative entropy}
\end{align}

Equation \eqref{eq: free energy relative entropy} now represents an optimization problem. 
Assume $\Qb^*$ is the optimal control distribution, and when the free energy is computed with respect to $S$ and $\Qb^*$, the free energy is minimized. This optimal free energy is upper bounded by the free energy computed from another distribution $\Qb$, summed with the KL-divergence between $\Qb$ and $\Qb^*$. In practice, MPPI assumes a form of the optimal distribution $\Qb^*$, which cannot be directly sampled from. Instead the KL-divergence term is utilized as an information theoretic metric that is used to drive a controlled distribution $\Qb$ closer to the optimal:

\begin{align}
    U^* = \argmin_U \text{KL}(Q^*||Q).
    \label{eq:control_obj_MPPI}
\end{align}

The authors in \cite{williams2017information} show that the solution to \eqref{eq:control_obj_MPPI} is equivalent to solving $\vu_t^* = \int q^*(V) \vv_t dV$. Where
\begin{align*}
    q^*(V) &= \frac{1}{Z}\exp \big( -\frac{1}{\lambda} S(V, \vx_0) \big) p(V)
\end{align*}

MPPI utilizes iterative importance sampling to approximate samples from the optimal distribution.

\begin{align}
    \vu_t^* &= \ExP{\Qb^*}{\vv_t}  \nonumber \\ 
    &= \ExP{\Qb}{\vv_t \frac{\rd \Qb^*}{\rd \Pb} \frac{\rd \Pb}{\rd \Qb} } \nonumber \\
    &= \ExP{\Qb}{\vv_t \exp \big( -\frac{1}{\lambda} S(V, \vx_0) \big) \frac{\rd \Pb}{\rd \Qb} } \nonumber
\end{align}
\textbf{Remark.} At this point, an obvious step would be to implement MPPI with the safety embedded dynamics \eqref{eq: safety embedded model}. This results in the Information Theoretic MPPI algorithm applied to the embedded barrier state control problem in \Cref{eq: bas optimal control problem}. The barrier state must be explicitly penalized in the cost function. This formulation allows MPPI to combine the optimization of task performance with safety, however, as we will see in \Cref{section: results}, some limitations exist. In general, there are scenarios where we lose the ability to explore if the dynamics are too close to an obstacle or in an undesirable region of state space.

In the next section, we go a step further and apply the safe controller to the importance sampling step. This has the effect of adding feedback with respect to the barrier state of the system, pushing samples away from unsafe regions. Additionally, this step circumvents the need to find tuning parameters to weight the cost of the barrier state versus the cost of the trajectory. The burden of solving the safe control problem falls onto a sub-optimization step that is purely focused on safety, while the MPC controller is tuned for performance.


\section{Safe Information Theoretic Model Predictive Control} \label{sec: Safe Information Theoretic Model Predictive Control }
\subsection{Safe Information Theoretic Measure} 

In this section, we will re-derive Information Theoretic Model Predictive Control with an alternative definition of the state-to-path cost function. This outline closely follows the derivation from \cite{williams2017information}. First, we define the state-to-path-cost function as
\begin{align}
\label{eq: safe path cost} S(V, \vx_0) = 
\begin{cases}
\phi(\vx_T) + \sum_{t = 0}^{T-1} q(\vx_t),& \vx \in \safeset, \\
\infty,& \vx \in \mathcal{C}^\text{C}.
\end{cases}
\end{align}
This cost function is applied to the following system,
\begin{align}
    \begin{bmatrix}
    \vx_{k+1} \\ \beta_{k+1} \end{bmatrix}
    &= 
    \begin{bmatrix}
    F(k, \vx_k, \vu_k + K_{\text{BaS}} \cdot \beta_k) \\ F^{\beta} (\vx_k, \beta_k, \vu_k+ K_{\text{BaS}} \cdot \beta_k)
    \end{bmatrix},
\end{align}
with \textit{a safety controller applied during the importance sampling}, see \Cref{fig: Quadrotor samples with feedback zoom}.

\begin{figure*} [th!]
    \centering
 \hspace{0mm}
  \subfloat[Unsafe Reference]{\includegraphics[trim={0 0 0 50},clip,width=0.30\linewidth]{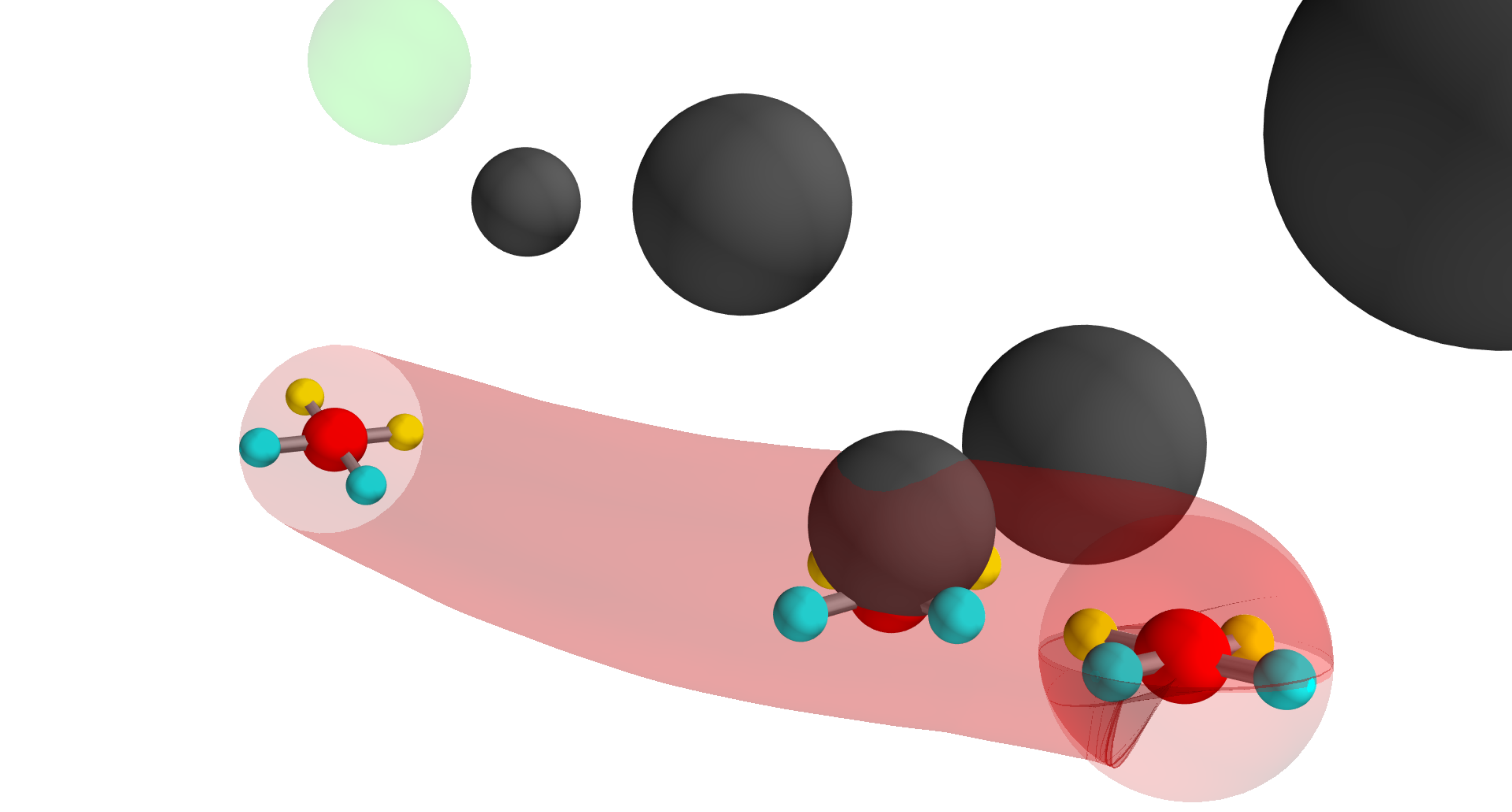}} \hspace{5mm}
    \subfloat[Corrected Reference]{\includegraphics[trim={0 0 0 120},clip,width=0.30\linewidth]{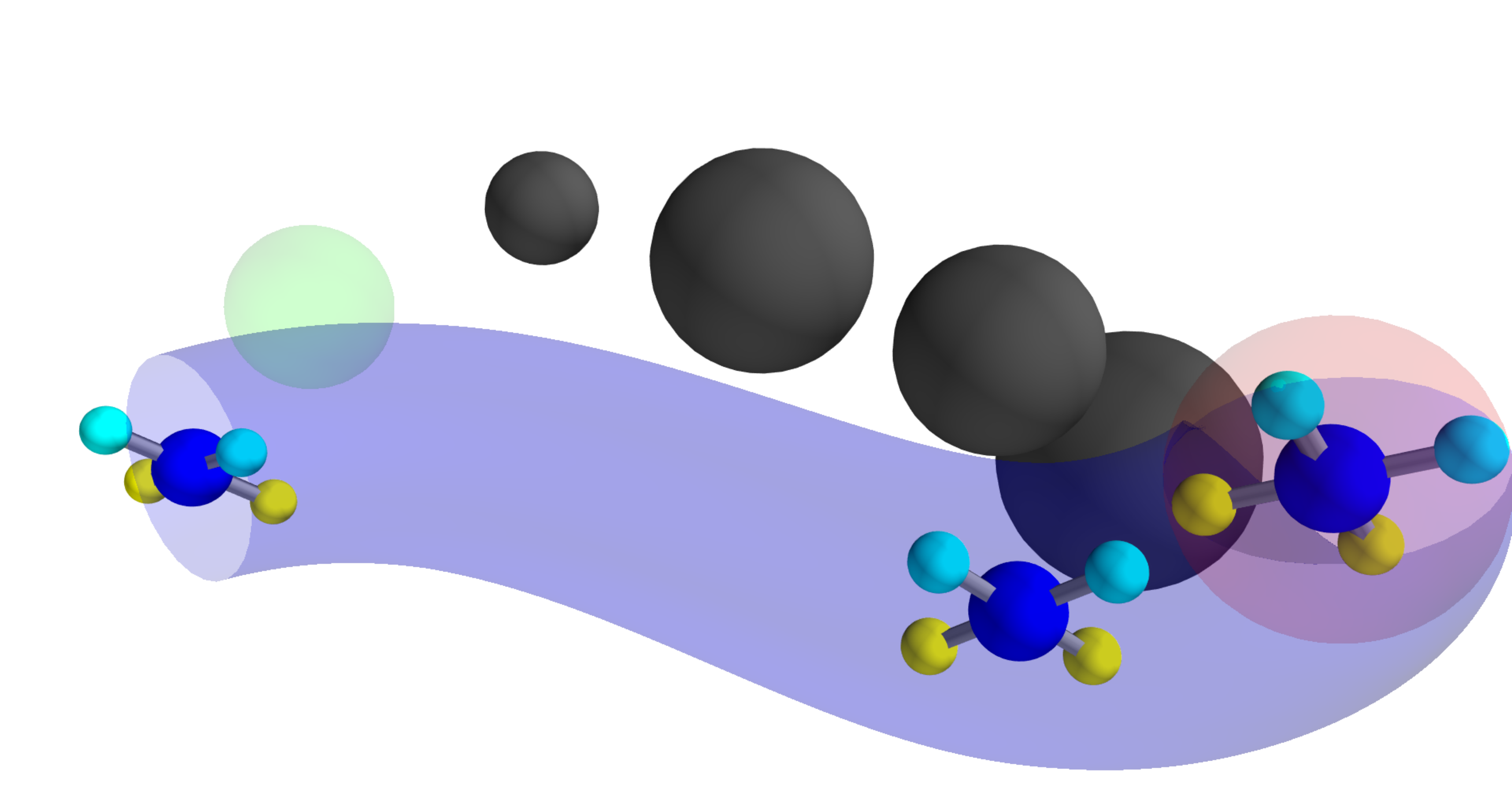}} \hspace{5mm}
  \subfloat[Samples under Barrier State Feedback]{\includegraphics[trim={0 0 0 50},clip,width=0.30\linewidth]{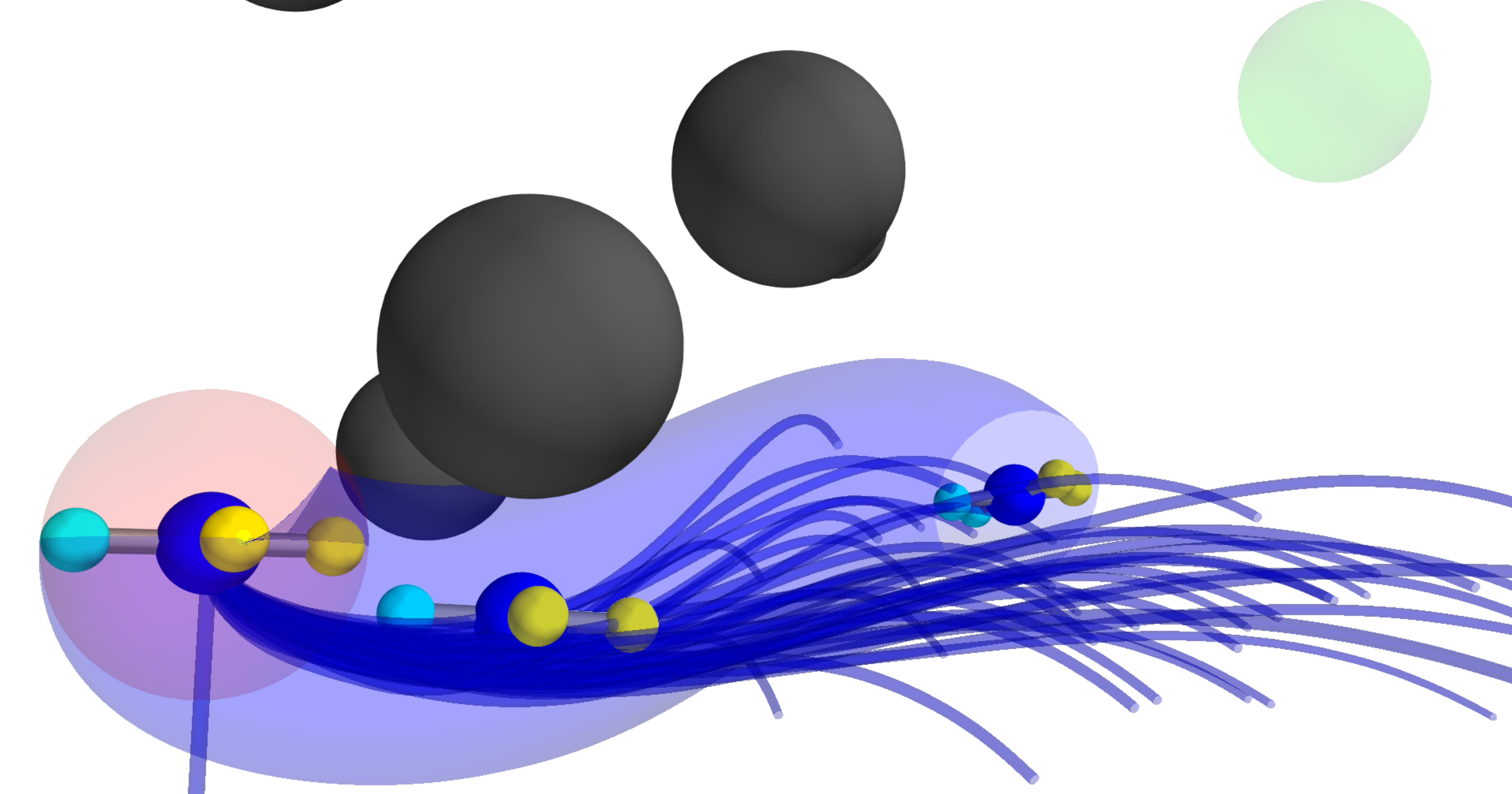}}
  \caption{The proposed importance sampling scheme, with a red unsafe reference trajectory in (a), a  blue corrected safe trajectory shown in (b), and blue safe samples under barrier state feedback shown in (c). Note how the samples attempt to curve around the obstacles.}
  \label{fig: Quadrotor samples with feedback zoom} 
\end{figure*}

\noindent \textbf{Remark} The state-to-path-cost function \eqref{eq: safe path cost} is not a function of $\beta$, and does not penalize the barrier state in any way. This separation of safety and performance enables the user to first tune an appropriate safety controller for a given tracking task, then design a path cost for the overall problem. The free energy under this state-to-path-cost is potentially infinite if the set of controls $V$ drives the system into the unsafe region $\mathcal{C}^C$.

Take note that the expectation in \eqref{eq: free energy} is a \textit{conditional expectation}, in this case being conditioned on the initial state $\vx_0$. We can split the measure $\Pb$ into two disjoint measures, one with control samples, that when combined with an initial condition, result in trajectories that are forward-invariant in $\safeset$ and another with trajectories that enter the unsafe region. In other words, the safe measure is parameterized by mean $\vu + K_{\text{BaS}} \cdot \beta$.
\begin{align}
 \Pb = \Pb_S \cup \Pb_U, \\
 \Pb_S \cap \Pb_U = \emptyset.
\end{align}
Using the additivity property of measures, we can then split the free energy into two terms. Using the fact that $S(V, \vx_0) = \infty$ when $V$ is sampled from $\Pb_U$, we see that the unsafe term goes to zero, since $\exp(-\infty) = 0$.
\begin{align} \label{eq: safe free energy}
&\mathcal{F}(S, \mathbb{P}, \vx_0, \lambda) = -\lambda \log \Bigg[ \ExP{\Pb_S}{\exp(-\frac{1}{\lambda} S(V,\vx_0))} \nonumber + \\
& \ExP{\Pb_U}{\cancelto{0}{\exp(-\frac{1}{\lambda} S(V, \vx_0))}} \Bigg]   = -\lambda \log \Bigg[ \ExP{\Pb_S}{\exp(-\frac{1}{\lambda} S(V, \vx_0))} \Bigg] 
\end{align}
We can upper bound the free energy $\mathcal{F}$ using Jenson's Inequality,
\begin{align}
    \mathcal{F} &= -\lambda \log \Bigg[ \ExP{\Qb_S}{\exp(-\frac{1}{\lambda} S(V, \vx_0))  \frac{ \rd \Pb_S}{ \rd \Qb_S}}  \Bigg] \nonumber \\
    &= -\lambda \log \Bigg[ \ExP{\Qb_S}{\exp(-\frac{1}{\lambda} S(V, \vx_0))  \frac{ \rd \Pb_S}{ \rd \Qb_S}} \nonumber \\ 
    &+ \ExP{\Qb_U}{\cancelto{0}{\exp(-\frac{1}{\lambda} S(V, \vx_0))  \frac{ \rd \Pb_S}{ \rd \Qb_S}}}\Bigg] \nonumber \\
    &\leq \ExP{\Qb_S}{ S(V, \vx_0)} + \lambda \ExP{\Qb_S}{\log \frac{ \rd \Qb_S}{ \rd \Pb_S}} \label{eq: safe free energy relative entropy}
\end{align}
\Cref{eq: safe free energy relative entropy} now represents a constrained optimization problem with the solution $\Qb^*$ achieving the lower bound in the free energy inequality.
\begin{lemma} Let $\frac{ \rd \Qb_S^*}{ \rd \Pb_S} = \frac{1}{\eta}\exp(-\frac{1}{\lambda} S(V, \vx_0)) $, with $\eta = \ExP{\Pb_S}{\exp(-\frac{1}{\lambda} S(V,\vx_0))}$. Then \eqref{eq: free energy relative entropy} reduces to an equality. \label{lemma: optimal safe distribution}
\end{lemma}
\begin{proof}
\begin{align}
    \mathcal{F} &\leq \ExP{\Qb_S}{ S(V, \vx_0)} + \lambda \ExP{\Qb_S}{\log \frac{ \rd \Qb_S^*}{ \rd \Pb_S}} \nonumber \\
    &= \ExP{\Qb_S}{S(V, \vx_0)} + \lambda \ExP{\Qb_S}{-\frac{1}{\lambda} S(V, \vx_0) - \log \eta} \nonumber \\
    &= \ExP{\Qb_S}{S(V, \vx_0)} - \ExP{\Qb_S}{S(V, \vx_0)} \nonumber \\
    &- \lambda \ExP{\Qb_S}{ \log \Big[ \ExP{\Pb_S}{\exp(-\frac{1}{\lambda} S(V,\vx_0))} \Big] }  \nonumber \\
    &= - \lambda \log \Big[ \ExP{\Pb_S}{\exp(-\frac{1}{\lambda} S(V,\vx_0))} = \mathcal{F}
\end{align}
\end{proof}

From \Cref{lemma: optimal safe distribution}, we see that as long as our likelihood ratio $\frac{ \rd \Qb_S^*}{ \rd \Pb_S}$ is proportional to $\exp(-\frac{1}{\lambda} S(V, \vx_0))$, we can use iterative importance sampling to estimate the optimal control distribution. In this case, the samples are taken from the safe distributions $\Pb_S$ and $\Qb_S$. 
Now the control objective is the following, and can be solved utilizing the same methods as regular MPPI \cite{williams2017information}.

\begin{align}
    U^* = \argmin_U \KL{\Qb^*_S}{\Qb_S}
    \label{eq: safe control_obj_MPPI}
\end{align}

\subsection{Safety Controlled Model Predictive Path Integral Control (SC-MPPI)} 
In this section we present an algorithm for safety controlled importance sampling in the context of MPPI. This framework is known as \ac{SC-MPPI}. The differences between this algorithm and traditional \ac{MPPI} are subtle, but important. The first difference is the computation of the Importance Sampling Sequence $U = \left(\vu_0 \dots \vu_{T-1} \right)$. The initial state of the system, along with an initial (potentially unsafe) importance sampling sequence are utilized to compute a safe importance sampling sequence $U_S$, along with any required parameters for the safe feedback controller $K_{\text{BaS}}$, see \Cref{fig: Quadrotor samples with feedback zoom}. In this work, we utilize Discrete Barrier State DDP \cite{almubarak2022safeddp} for the embedded safety controller. Note that this feedback controller is only valid within a domain of the nominal trajectory. The second difference appears in the form of safe sampling, where the embedded safety controller is utilized to perform feedback on the barrier state and the barrier state alone while forward sampling trajectories. The sampling procedure is summarized in \Cref{alg: safe AIS}, and the full control algorithm is summarized in \Cref{alg: SC-MPPI}.
\begin{algorithm}[htbp]
\footnotesize
\SetKwInOut{Input}{Given}
\Input{
      $F$, $q$, $\phi$, $\Sigma$, $R$, $R_{\text{fb}}$: Dynamics, Cost function parameters\;
      $T$, $N$: Sampling parameters\;
      $\lambda, \beta, \nu, j$: Temperature and control cost smoothing parameter, nominal feedback scale, barrier state index\;
      }
\SetKwInOut{Input}{Input}
\Input{
    $\vx_0$, $U_S$, $K_{\text{BaS}}$: Current state, Safe IS sequence, safe feedback controller\;
}
\BlankLine
\For{$n \leftarrow 1$ \KwTo $N$}{
    $\vx \leftarrow \vx_0$;\quad $S^n \leftarrow 0$\;
    Sample $\mathcal{E}^n = \left( \epsilon_0^n \dots \epsilon_{T-1}^n \right), ~\epsilon_k^n \in \mathcal{N}(0, \Sigma)$\;
    \For{$k \leftarrow 0$ \KwTo $T-1$}{
        \If{$k > 0$}{
        $S^n \pluseq q(\vx) $\;
        $S^n \pluseq \frac{\lambda (1-\alpha)}{2} \big( k_{\text{fb}}^{\rT} R_{\text{fb}}\Sigma^{-1} k_{\text{fb}} + (\vu_k + 2\epsilon_k^n)^{\rT} R\Sigma^{-1} \vu_k \big)$\;   
        }
        $\bar{\vx} \leftarrow \vx$\;
        $\bar{\vx}(j)\leftarrow 0$\; 
        $k_{\text{fb}} \leftarrow \nu \cdot K_{\text{BaS}}(\vx, \bar{\vx})$\;
        $\vx \leftarrow F\left(\vx, \vu_k + \epsilon_k^n + k_{\text{fb}} \right)$\;
    }   
    $S^n \pluseq \phi(\vx)$
}
\Return{$\mathbf{S} = \left( S^0 \dots S^N \right)$, $\mathcal{E}^n$}\;
\caption{Safety Controlled Importance Sampler (SCIS)} 
\label{alg: safe AIS}
\end{algorithm}

\begin{algorithm}[htbp]
\footnotesize
\SetKwInOut{Input}{Given}
\Input{$F, K_{\text{BaS}}, P$: Dynamics, Safety Controller\;
$P$, $\lambda$:  Maximum iterations, Temperature\;}
\SetKwInOut{Input}{Input}
\Input{
    $\vx_0$, $U$: Current state, Initial IS sequence\;}
\BlankLine
\For{$i \leftarrow 0$ \KwTo $P$}{
    $\big( U_S, K_{\text{BaS}}\big) \leftarrow \text{computeSafeFeedback}(\vx_0, U)$\;
    $\mathbf{S},~ \mathcal{E}^n \leftarrow \text{safetyControlledImportanceSampler}(\vx_0, U_S, K_{\text{BaS}})$\;
    $\bar{S} = \text{min}(\mathbf{S})$\;
    \For{$n = 1$ \KwTo $N$}{
        $S^n \leftarrow \text{exp}\big( -\frac{1}{\lambda} (S^n - \bar{S})  \big)$
    }
    $\eta \leftarrow \sum_{n=1}^N S^n$\;
    \For{$n = 1$ \KwTo $N$}{
        $w^n \leftarrow S^n / \eta$
    }
    \For{$k \leftarrow 0$ \KwTo $T-1$}{
    $\vu_k^* = 0$\;
        \For{$n = 1$ \KwTo $N$}{
            $\vu_k^* \pluseq w^n \cdot (\vu_k + \epsilon_k^n)$\;
        }
    }
    $U = \left( \vu_0^* \dots \vu_{T-1}^* \right)$
}
\Return{$U $}\;
\caption{\acf{SC-MPPI}} 
\label{alg: SC-MPPI}
\end{algorithm}

\ 
\section{Results} \label{section: results}
We now test the proposed algorithm \ac{SC-MPPI} against \textit{vanilla} \ac{MPPI}, under the barrier state dynamics. The algorithms were tested on a Dubins vehicle and a multirotor system in simulation. All experiments were run on an Intel i7-12700K with 32 GB of RAM and a NVIDIA RTX 3080 Ti GPU. The cost functions used in the experiments have the following form:
\begin{align*}
J_{\text{DDP}} &= \sum_{k=0}^{T-1} \big(\vx_k^\text{T} Q \vx_k +  \vu_k^\text{T} R \vu_k  + q_{\beta}\beta_k^2  \big) + \vx_T^\text{T} \Phi \vx_T 
\\
J_{\text{MPPI}} &= \sum_{k=0}^{T-1} \big(\vx_k^\text{T} Q \vx_k + \frac{\lambda (1-\alpha)}{2} \vu_k^\text{T} R \Sigma^{-1} \vu_k  \\&+ q_{\beta}\beta_k^2  \big) + \vx_T^\text{T} \Phi \vx_T 
\\
J_{\text{SC-MPPI}} &= \sum_{k=0}^{T-1} \big(\vx_k^\text{T} Q \vx_k + \frac{\lambda (1-\alpha)}{2} \vu_k^\text{T} R \Sigma^{-1} \vu_k  \\&+ k_{\text{fb}}^{\rT} R_{\text{fb}}\Sigma^{-1} k_{\text{fb}}  \big) + \vx_T^\text{T} \Phi \vx_T 
\end{align*}
where $k_{\text{fb}} = K_{\text{BaS}} \cdot \beta_k$. 
Note, an important aspect of \ac{SC-MPPI} is the computation of the safe controller along the importance sampling trajectory. MPPI might generate an unsafe reference trajectory and hence using classical barrier functions with DBaS-DDP is not feasible as it requires safe initializations (see \cite{almubarak2022safeddp}). Relaxed barrier functions \cite{hauser2006barrier,feller2016relaxed} would allow such scenarios on the other hand. In other words, relaxed barrier functions allow DDP to converge to a solution even if the majority of the reference trajectory was unsafe in addition to ensuring numerical stability.




\subsection{Dubins Vehicle}
As a proof of concept of the proposed algorithm, a Dubins vehicle must navigate a cluttered environment. We set up a \textit{dense} navigation problem such that the vehicle should narrowly move between the obstacles given its size (the vehicle's radius is $0.2$ units). The experiment's details and the details of the implementation are provided in \autoref{App subsec: Dubins Dynamics}. 

\begin{figure} 
    \centering
  \subfloat[]{\includegraphics[trim={1cm 1cm 1cm 1cm},clip,width=0.75\linewidth]{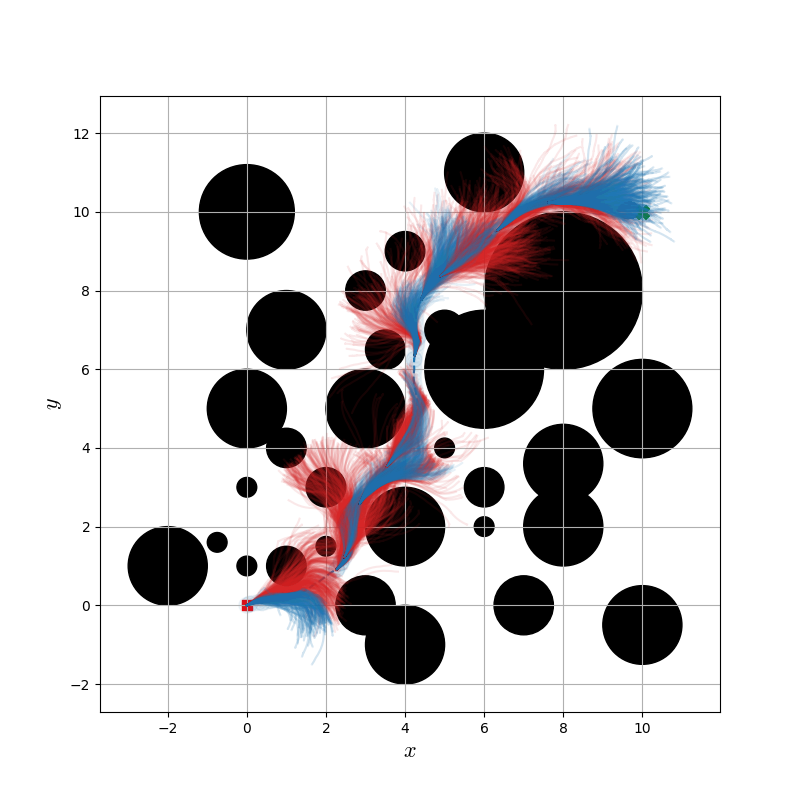}} \\ \vspace{-0.5mm}
  \subfloat[]{\includegraphics[trim={1cm 1cm 1cm 1cm},clip,width=0.75\linewidth]{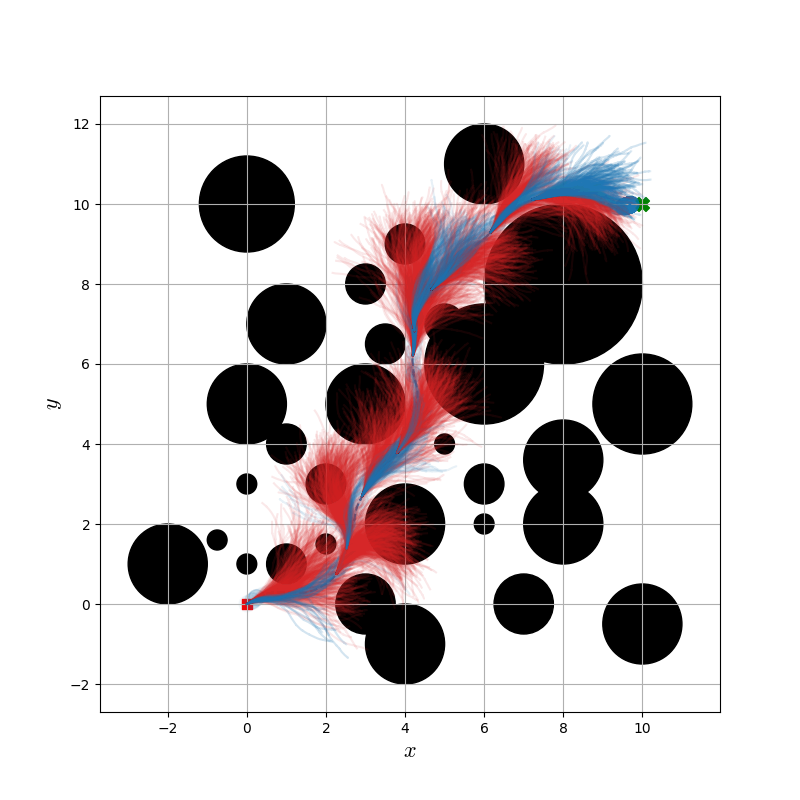}} \\ \vspace{-0.5mm}
  \caption{Dubins vehicle samples visualized with safe samples in blue, and unsafe samples in red. The red square is the start position, and the green X is the target location. The top figure is \ac{SC-MPPI} (with safety feedback on the samples) and the bottom one is \ac{MPPI}. It can be seen that \ac{SC-MPPI} samples are deflected, moving away from the obstacles projecting further safe exploration in addition to be close-packed between obstacles. }
  \label{fig: dubins sampling} 
\end{figure}

Implementation of the proposed algorithm, safety controlled MPPI, is shown in \autoref{fig: dubins sampling} (a), and vanilla MPPI with barrier in the cost is shown in \autoref{fig: dubins sampling} (b). To validate the idea of barrier states feedback in importance sampling, we show the two algorithms' samples (512 samples per step) at select time instances in the environment. Safe samples are shown in blue while unsafe samples are shown in red. Note that the environment is purposefully challenging to navigate,and as a result, a small perturbation can result in unsafe samples. As shown in \autoref{fig: dubins sampling}, the proposed algorithm \ac{SC-MPPI}, results in the samples to be deflected away from the obstacles, due to the barrier state feedback, effectively encouraging safe exploration. On the other hand, vanilla MPPI's samples are agnostic to the constraints and thus are distributed around the nominal trajectory in a parabola-like shape. In addition, it can be observed that vanilla MPPI's samples collide with more obstacles while \ac{SC-MPPI}'s samples stop before the obstacles most of the time. Furthermore, \ac{SC-MPPI} has more safe samples (blue) than vanilla MPPI.

Next, we provide detailed numerical comparisons between the algorithms for the multirotor example. 


\ 
\subsection{Multirotor}
The environment for this task has 19 obstacles of various sizes, and the system (with radius $1.5 m$) must navigate through this dense field, pictured in \autoref{fig: quadrotor sampling}. Safety violation is determined by collision of the system into any of the obstacles for a single time instant. Task completion is defined as entering within a $0.5 m$ radius of the desired final position. Two experiments are shown here, highlighting differences in performance which emerge from tuning, control variance, and problem horizon for both \ac{MPPI} and \ac{SC-MPPI}. We also compare the two algorithms against \acf{DBaS} embedded \acf{MPC-DDP}. The parameters for each of the controllers are provided in the Appendix (\Cref{App subsec: quad control params}).

\begin{figure*} [h]
    \centering
 \hspace{0mm}
  \subfloat[DDP]{\includegraphics[trim={20 30 60 20},clip,width=0.38\linewidth]{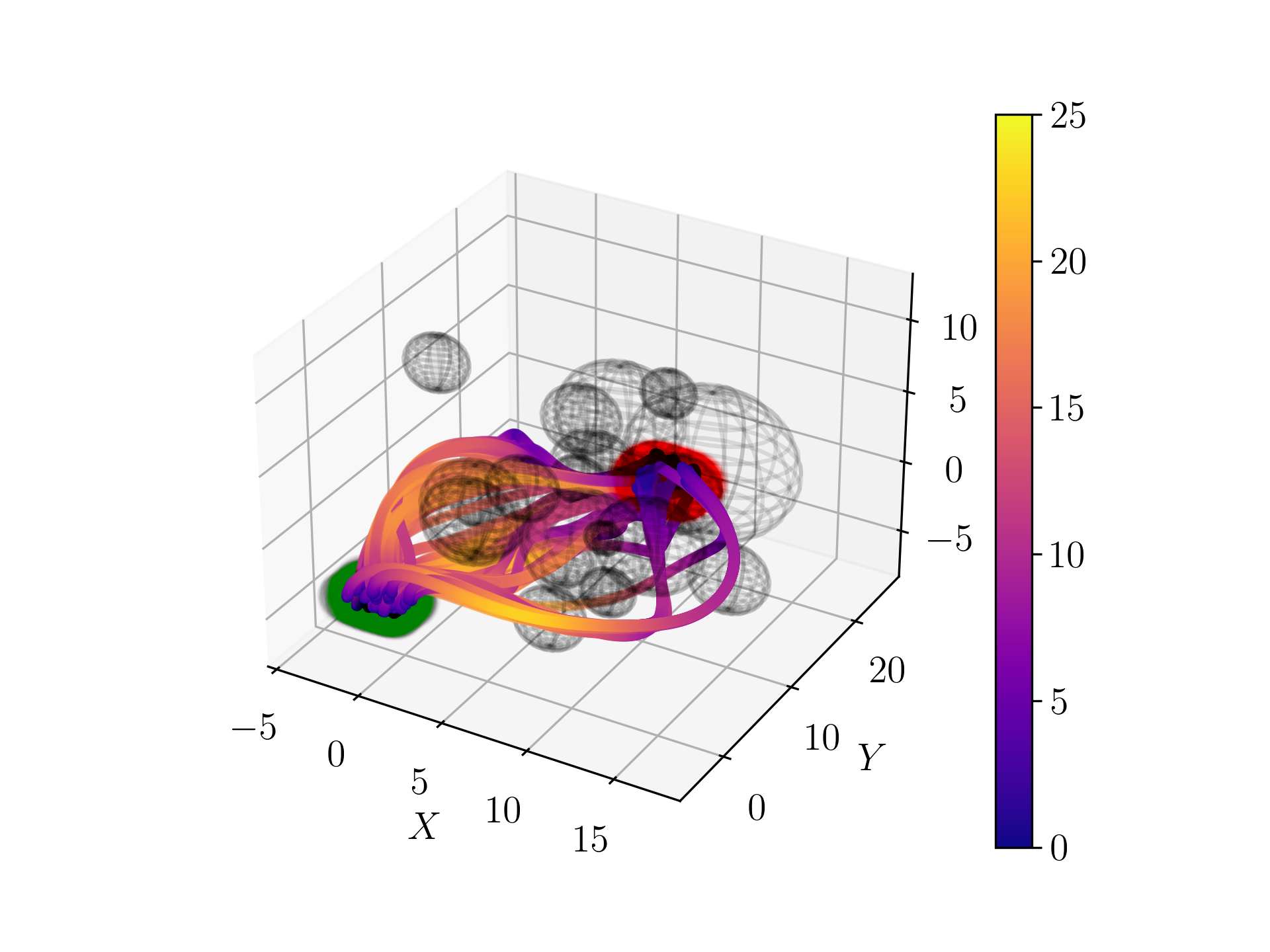}} \hspace{-15mm}
    \subfloat[MPPI]{\includegraphics[trim={20 30 60 20},clip,width=0.38\linewidth]{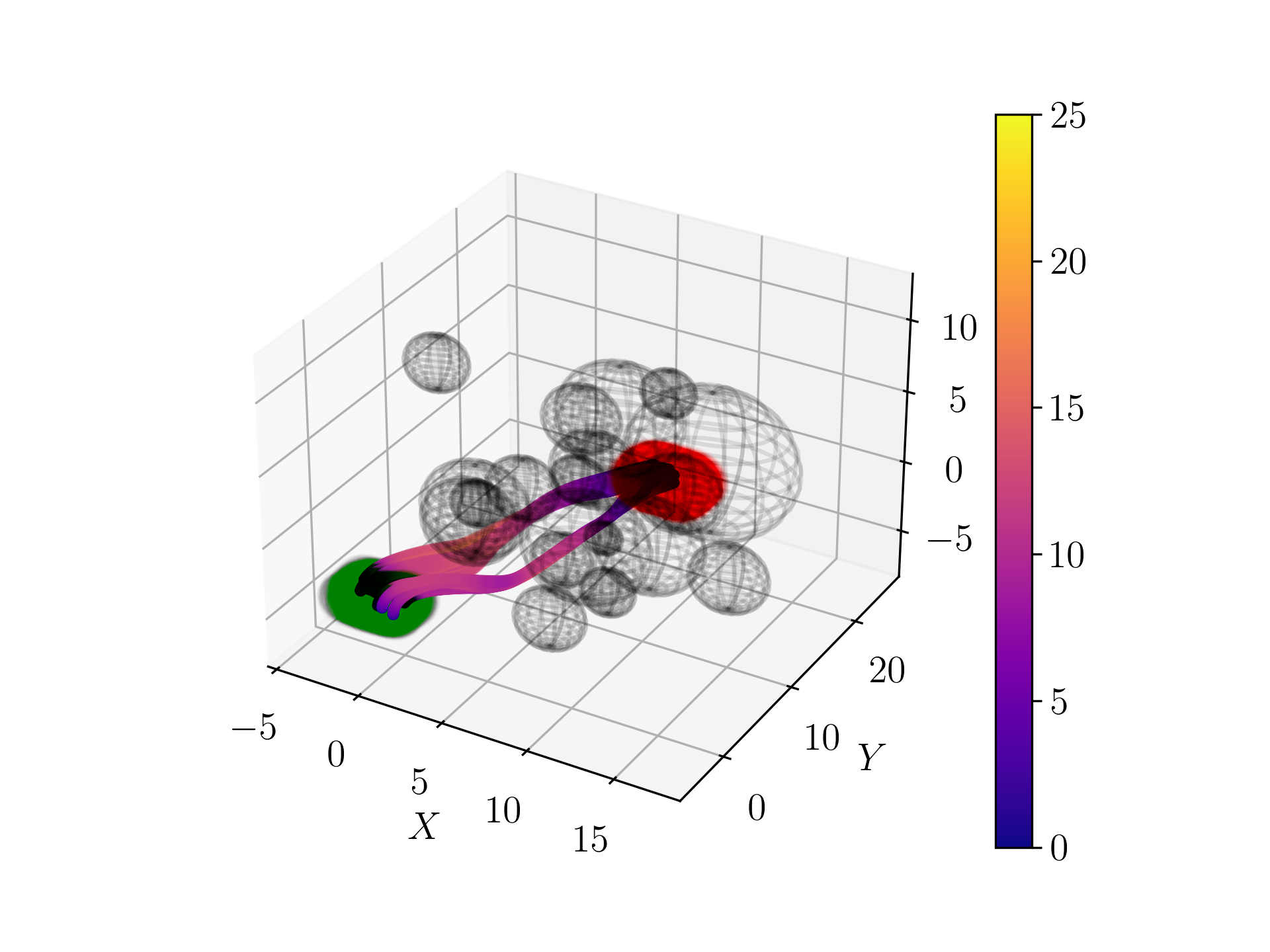}} \hspace{-15mm}
  \subfloat[SC-MPPI]{\includegraphics[trim={20 30 60 20},clip,width=0.38\linewidth]{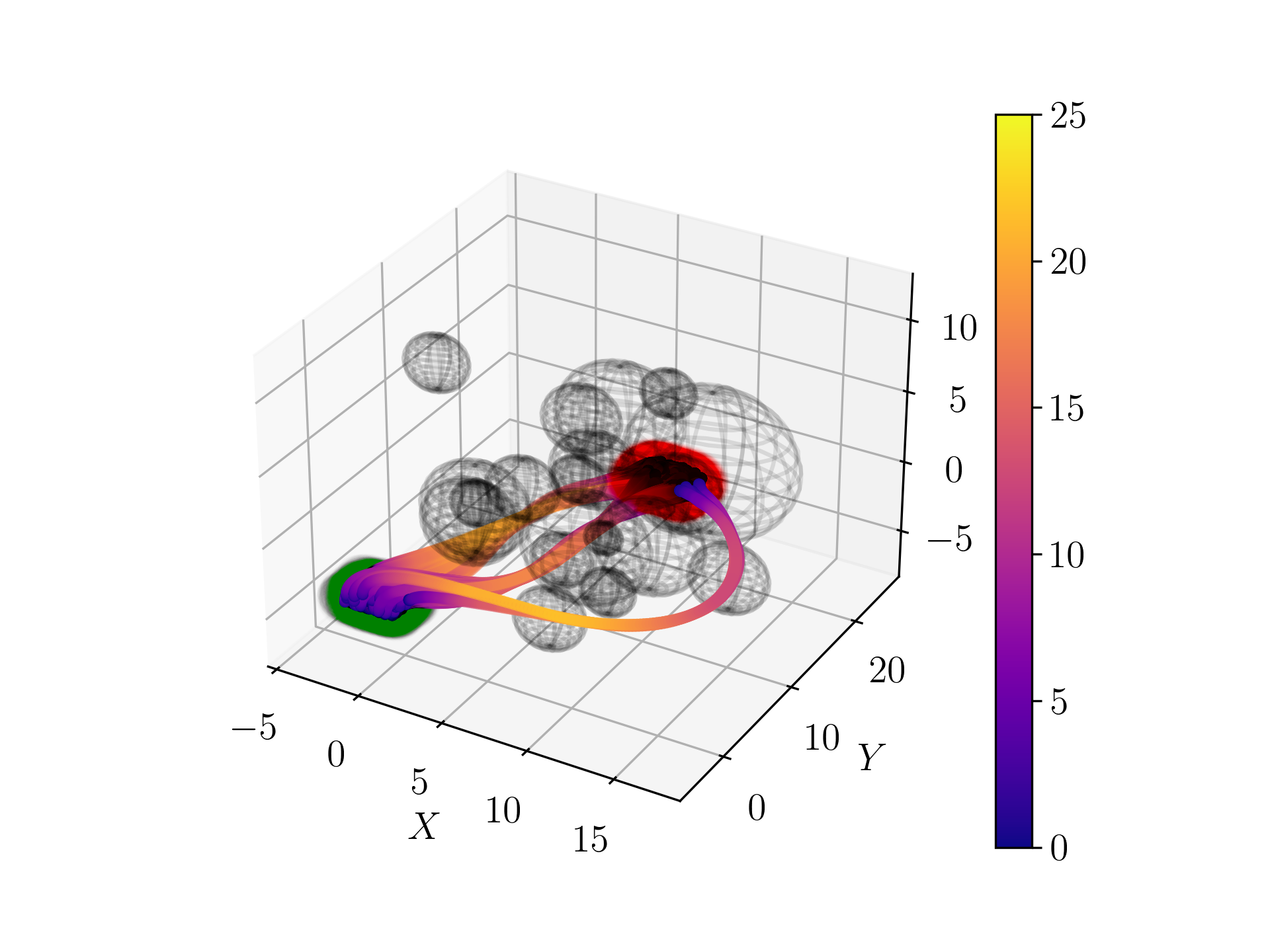}}
  \caption{\textbf{Experiment 1} Quadrotor trajectories visualized for each algorithm. The trajectories are colored based on velocity. (a) DDP, note how the solutions span more of the state space. (b) MPPI, the solutions focus on the center of the state space due the low control exploration variance. (c) \ac{SC-MPPI}, the solutions can explore the state space better due to the influence of the underlying safety controller, while maintaining better statistics \Cref{table: quad small variance}. The lower exploration variance hinders \ac{MPPI}, while \ac{SC-MPPI} is able to leverage the safety controller to explore more of the state space.}
  \label{fig: quadrotor trajectories small variance} 
\end{figure*}

\textbf{Experiment 1:} For the first experiment, the problem horizon is set to \textit{3 seconds}, with control variances for \acf{MPPI} and \acf{SC-MPPI} set to $\sigma = [5.0, 5.0, 5.0, 15.0 ]$. Trajectories for each of the trials are visualized in \Cref{fig: quadrotor trajectories small variance}, with each trajectory colored by velocity. The effect of lower control exploration variance is clear for \ac{MPPI}, where the trajectories that both maintain safety and complete the task pass almost directly through the dense obstacle field. Under the given tuning parameters, DDP has a high velocity but also high safety violation percentage at $36.32\%$ (\Cref{table: quad small variance}). In contrast, MPPI has a low completion percentage, but an even lower task completion rate when compared to DDP. \ac{SC-MPPI} outperforms both algorithms in safety, completion, RMS error, and task completion time under this task. 

\begin{table}[h!]
\scriptsize
\caption{\textbf{Experiment 1} Multirotor Dense Navigation Statistical Trial}
\centering
 \begin{tabular}{l |c |c |c} 
  & \ac{DDP} & \ac{MPPI} & \ac{SC-MPPI} \\ 
 \hline\hline
 Compute Time (ms)      & $16.56  \pm (1.56)$     & $\mathbf{3.16 \pm (0.06)}$        & $6.15 \pm (0.61)$ \\ 
  \hline
 Safety Violation \%  & $36.32\%$              & $1.05\%$        & $\mathbf{0.84\%}$\\
  \hline
 Task Completion \%   & $16.63\%$              &  $4.32\%$               & $\mathbf{68.84\%}$ \\
  \hline
  Completion Time (s)   & $2.30 \pm (0.33)$      & $2.93 \pm (0.05)$        & $\mathbf{2.11 \pm (0.15)}$ \\
 \hline
 Position RMSE (m) & $0.62 \pm (0.18)$      & $0.46 \pm (0.03)$        & $\mathbf{0.14 \pm (0.14)}$ \\
 \hline
  Avg Velocity (m/s)    & $\mathbf{7.87 \pm (0.97)}$      & $6.30 \pm (0.27)$        & $6.78 \pm (0.44)$ \\
  \hline
  Max Velocity (m/s)    & $\mathbf{18.15 \pm (2.20)}$     & $14.21 \pm (1.46)$       & $17.15 \pm (1.67)$ \\
 \hline
 \end{tabular}
 \label{table: quad small variance}
\end{table}

The task completion issues for MPPI are likely due to the lower control exploration variance, as well as the limited time horizon for the problem. MPPI takes more time to find a solution around the obstacles, and does not have enough time enter the completion radius. DDP has large average and maximum speeds, but typically takes a longer path around all the obstacles. Under the same problem, \ac{SC-MPPI} has a safe sample rate of $38.54\%$, and \ac{MPPI} has a safe sample rate of $37.20\%$, when safe samples are averaged across all trajectories and timesteps. While the difference in the number of safe samples is quite small, the performance margin is quite large. The reasoning behind the low safe sample percentage is likely due to the density of the obstacle course, and the time limitation for the task. Since the task attempts to send the quadrotor through the field in under 3 seconds, most trajectories will impact the obstacles. In \Cref{fig: quadrotor sampling}, we can observe the differences in \ac{MPPI} and \ac{SC-MPPI} sampling for the quadrotor for Experiment 1 and directly see the variations in safe versus unsafe samples for the two sampling-based algorithms. The safe underlying controller for \ac{SC-MPPI} allows the reference trajectories to move closer towards the goal, and the barrier state feedback clearly forces the samples away from the obstacles.
\begin{figure} [h!]
    \centering
  \subfloat[]{\includegraphics[trim={0.0cm 0.0cm 1.5cm 1.0cm},clip,width=0.5\linewidth]{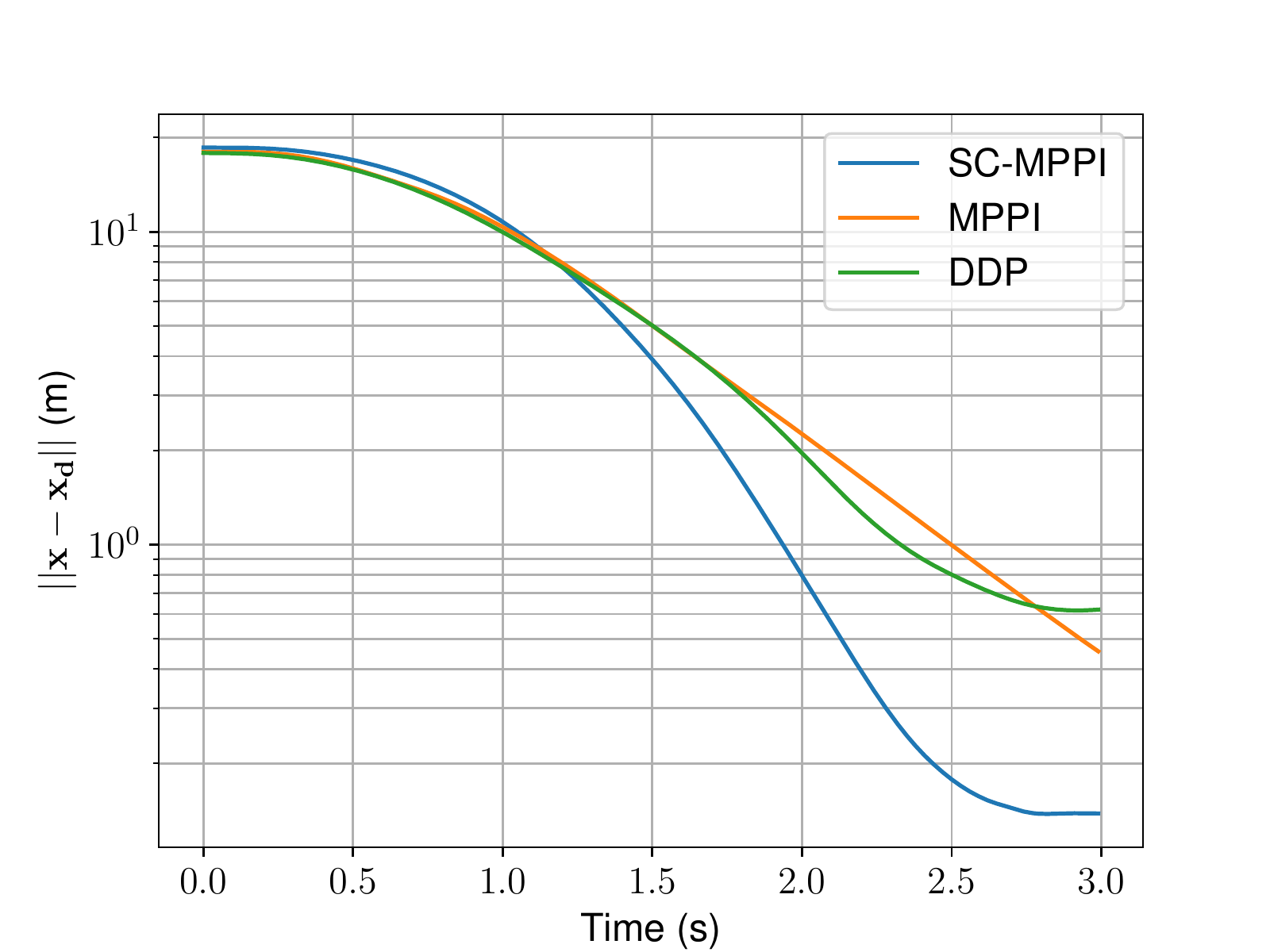}}
  \subfloat[]{\includegraphics[trim={0.0cm 0.0cm 1.5cm 1.0cm},clip,width=0.5\linewidth]{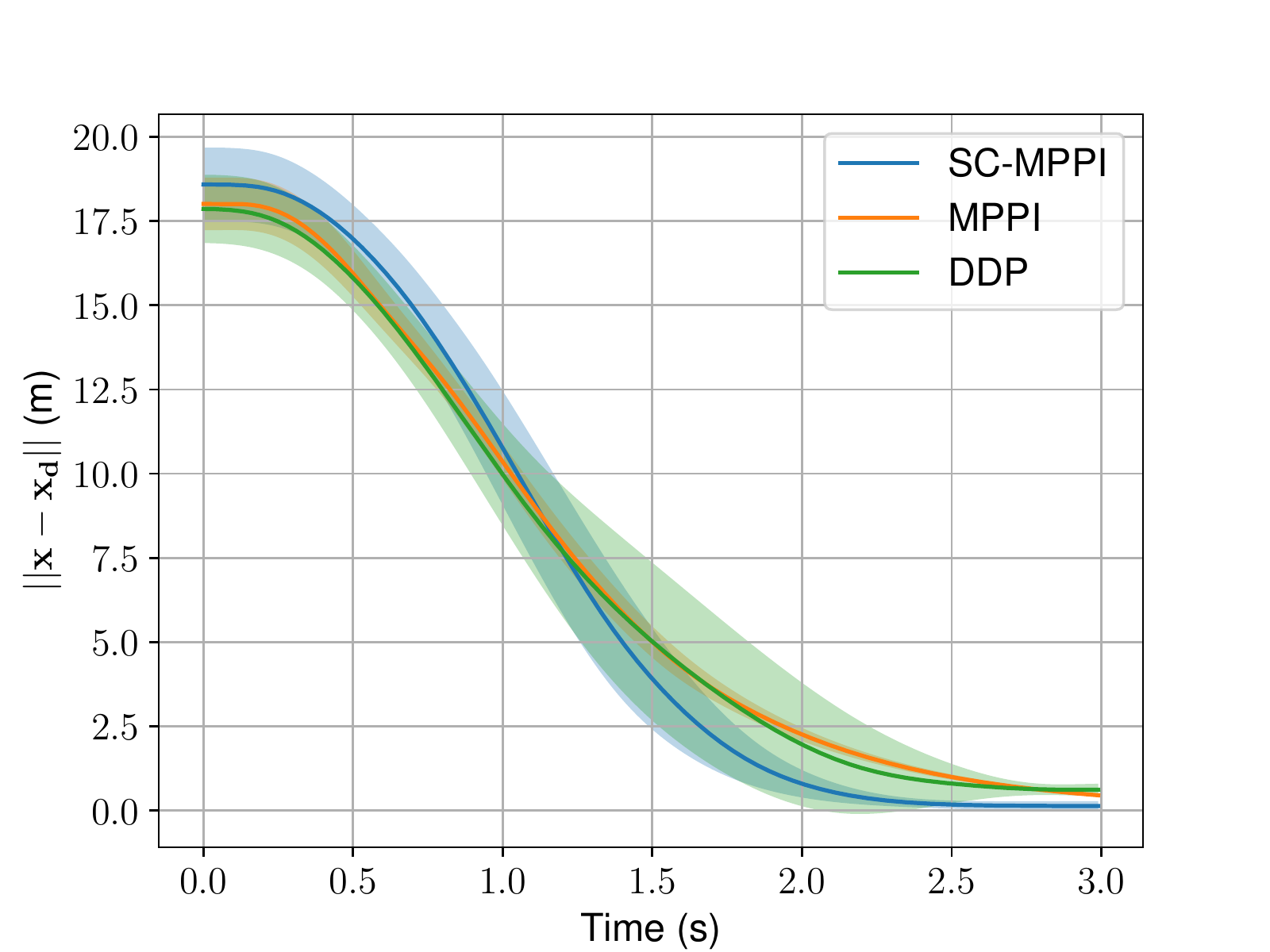}}
  \\
    \subfloat[]{\includegraphics[trim={0.0cm 0.0cm 1.5cm 1.0cm},clip,width=0.5\linewidth]{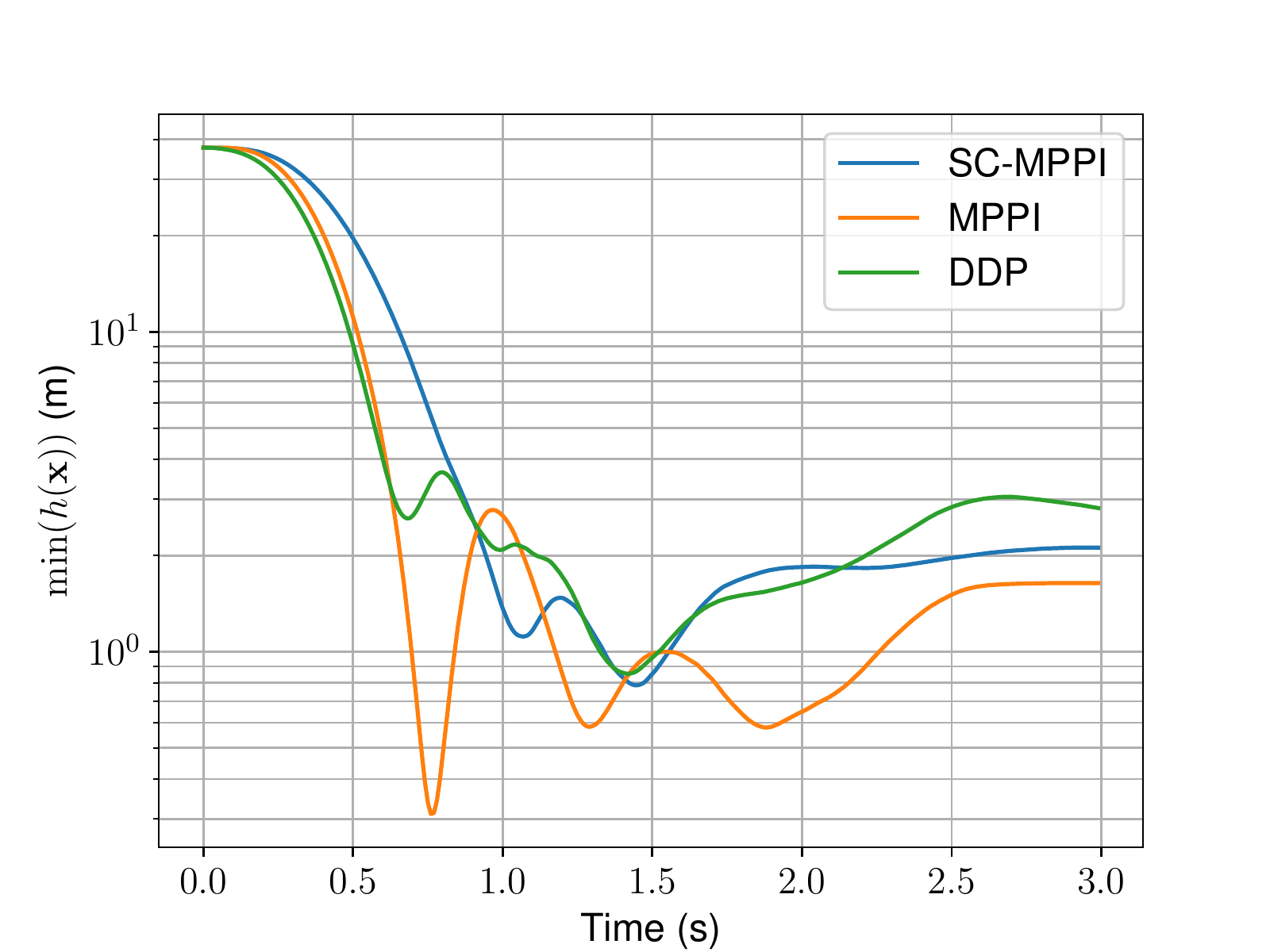}}
  \subfloat[]{\includegraphics[trim={0.0cm 0.0cm 1.5cm 1.0cm},clip,width=0.5\linewidth]{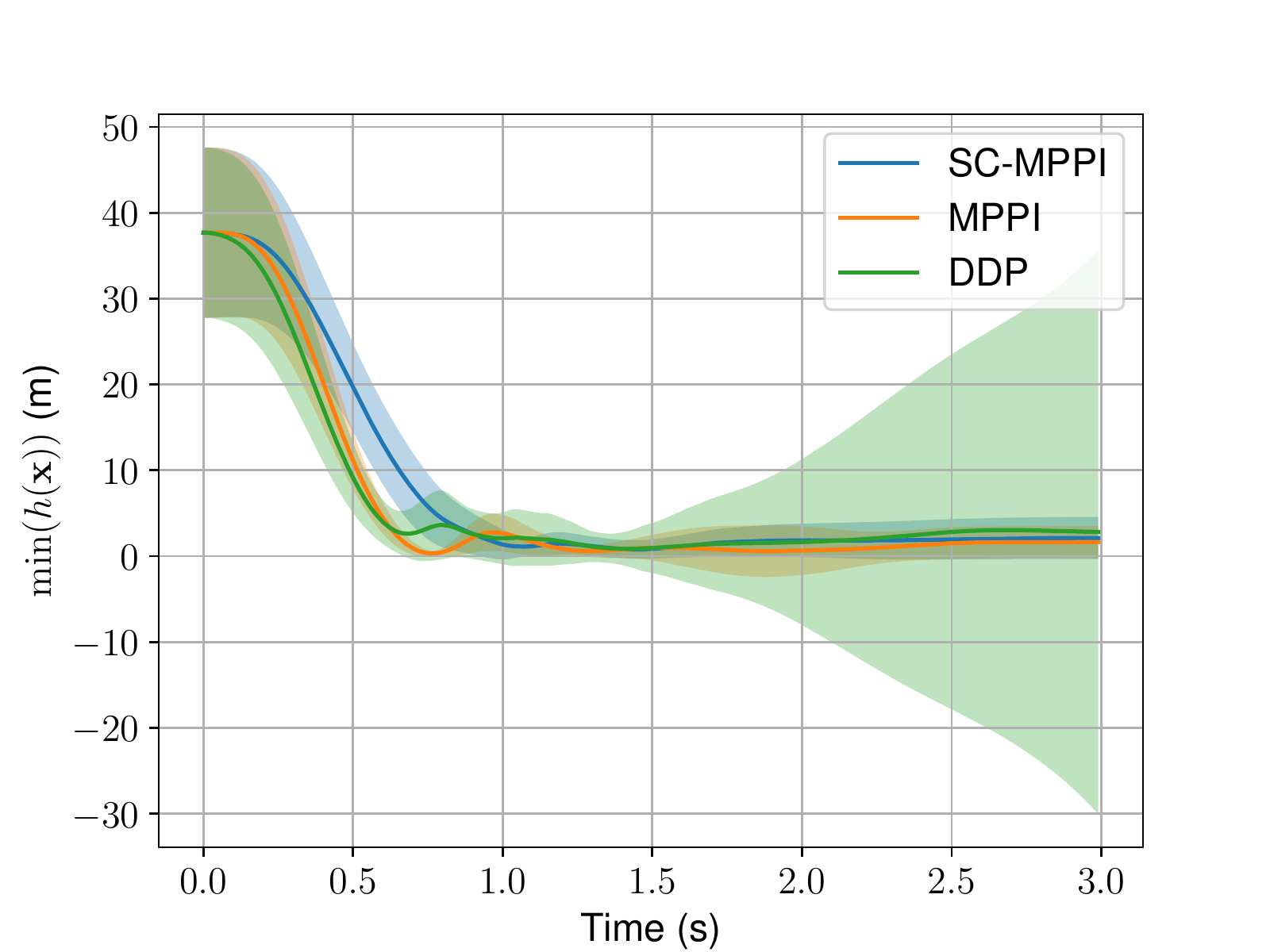}}
  \caption{In the top row is the average distance from the target for the compared algorithms in the logarithmic scale (left) and in the normal scale (right). In the bottom row is the average (over all safe, but potentially incomplete runs) minimum distance to the obstacles in the logarithmic scale (left) and the normal scale (right). Clearly, \ac{SC-MPPI} achieves more reachability to the target. We hypothesize that the reason is due to the encouraged safe exploration since the target is right behind an obstacle. Additionally, it maintains a further distance from the obstacles compared to MPPI.}
  \label{fig: quadrotor distance_results_experiment_1} 
\end{figure}

\begin{figure*} [th!]
    \centering
 \hspace{0mm}
  \subfloat[DDP]{\includegraphics[trim={20 30 60 20},clip,width=0.38\linewidth]{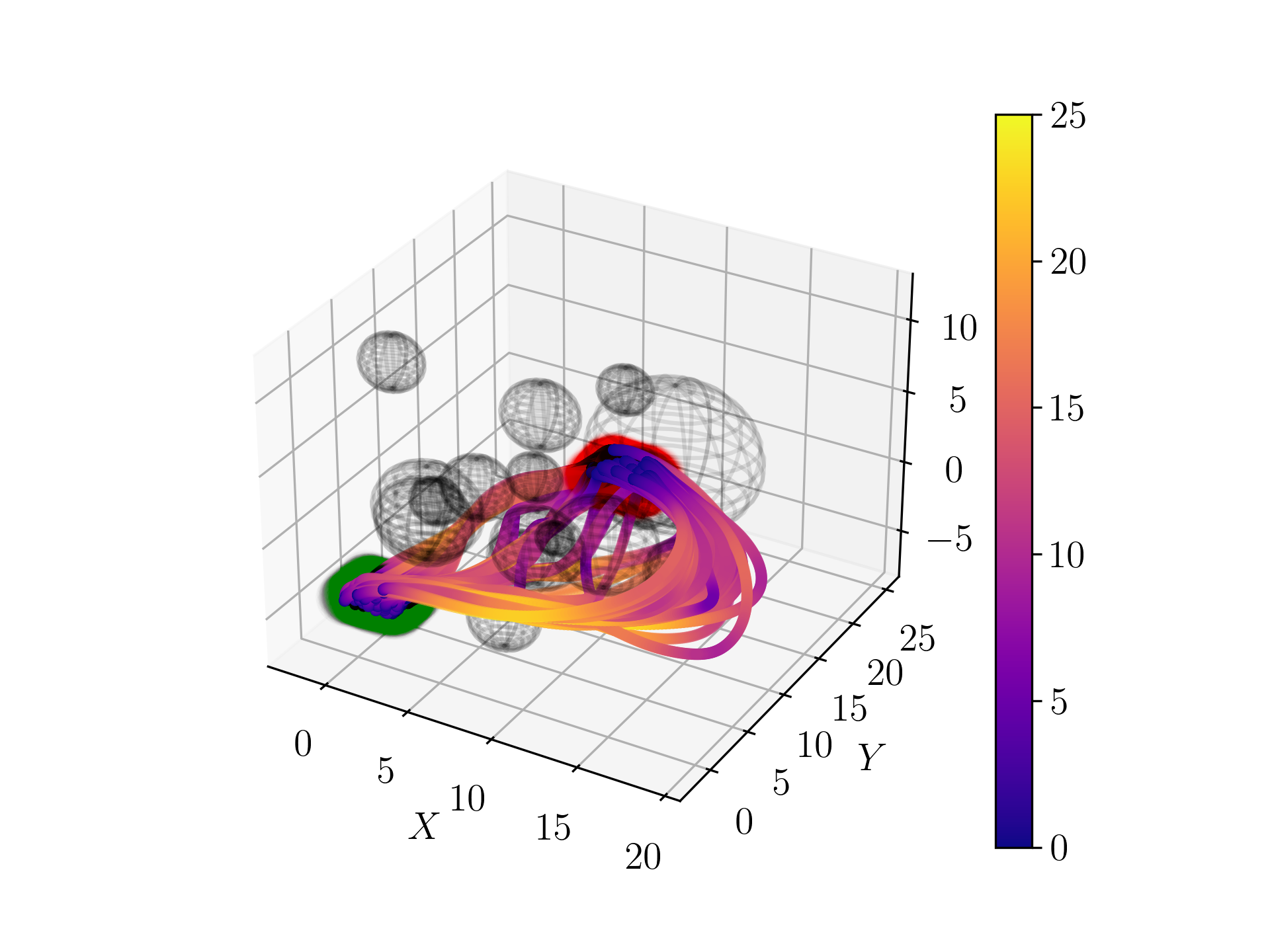}} \hspace{-15mm}
    \subfloat[MPPI]{\includegraphics[trim={20 30 60 20},clip,width=0.38\linewidth]{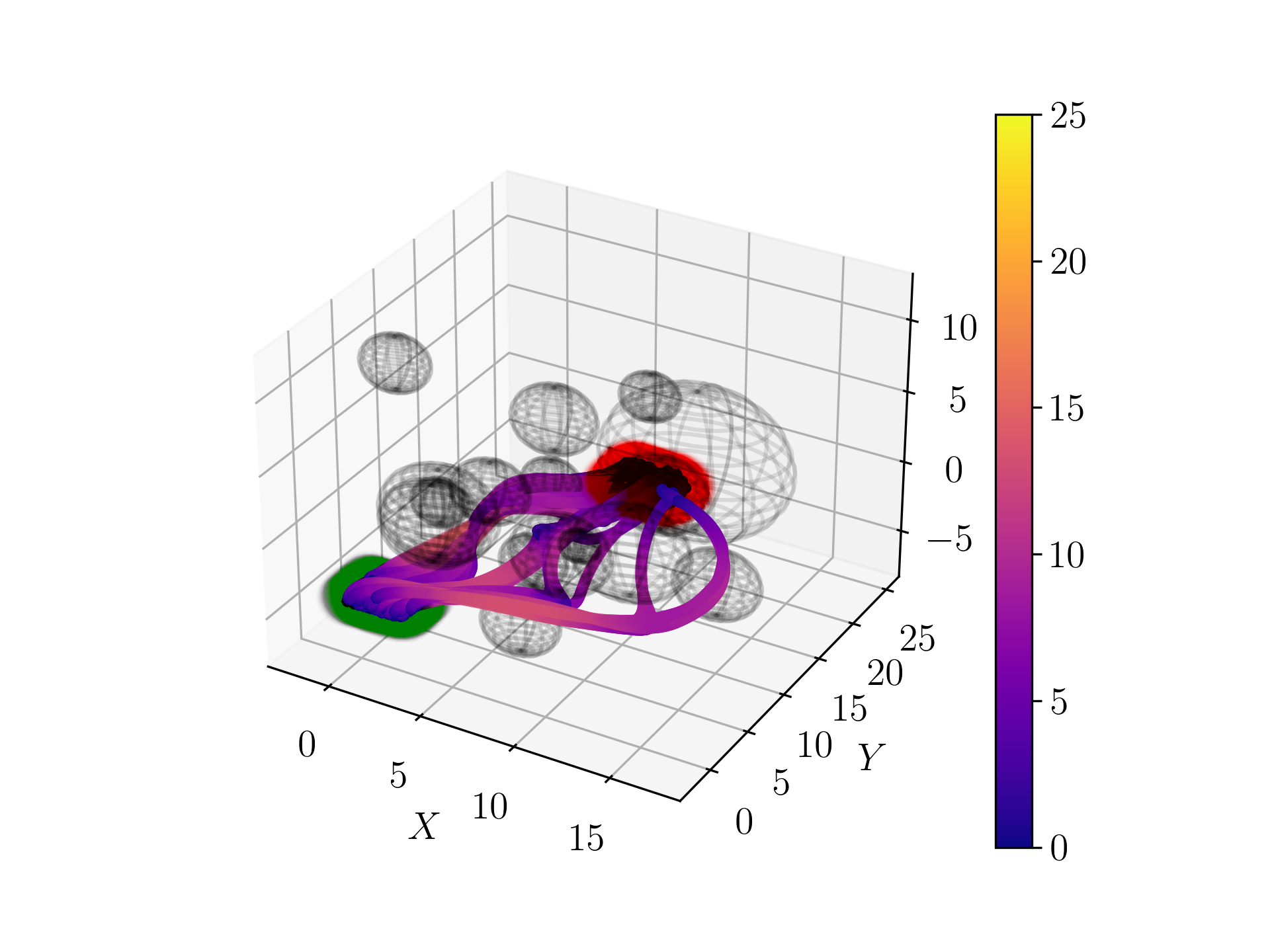}} \hspace{-15mm}
  \subfloat[SC-MPPI]{\includegraphics[trim={20 30 60 20},clip,width=0.38\linewidth]{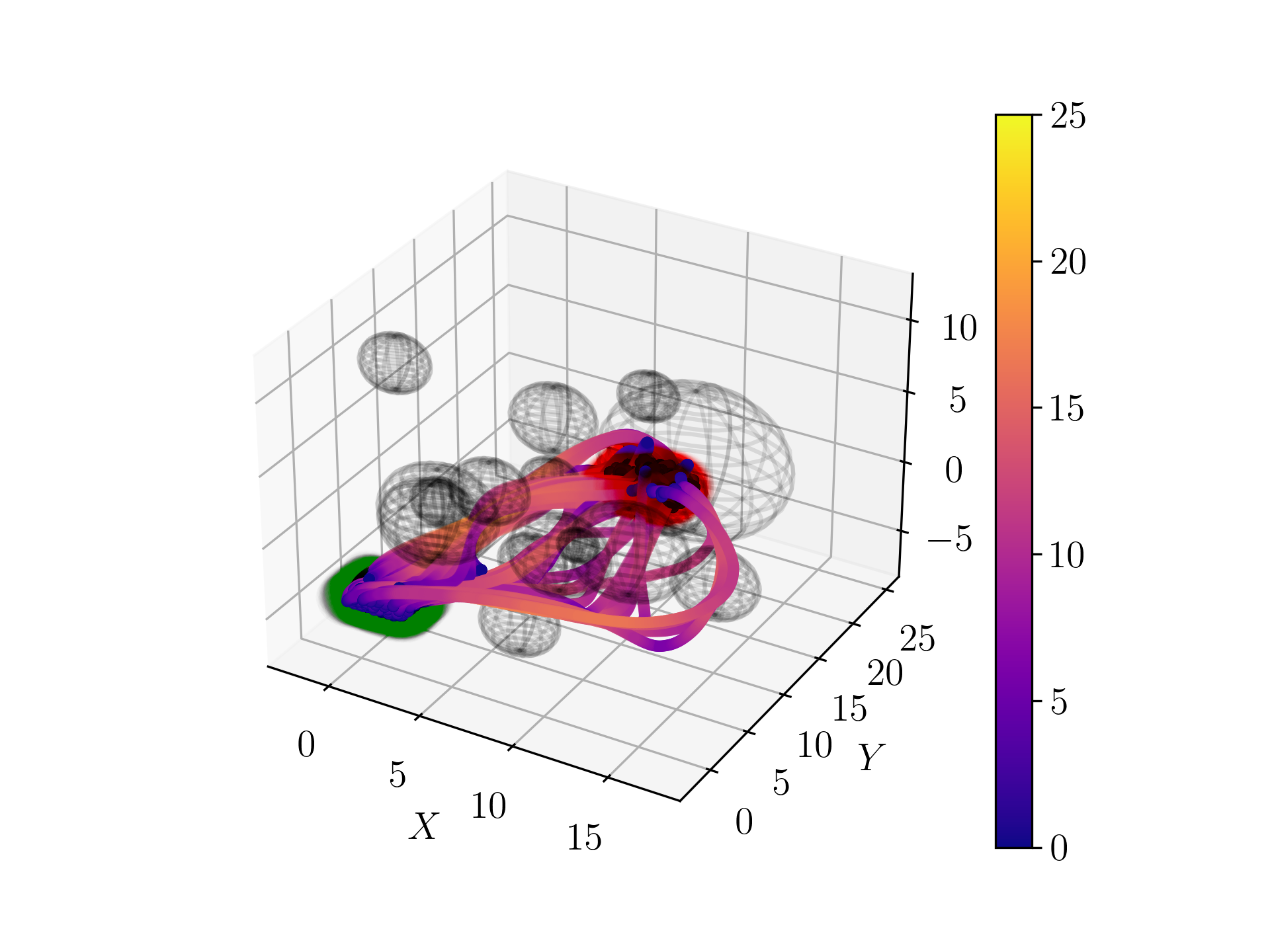}}
  \caption{\textbf{Experiment 2} Quadrotor trajectories visualized for each algorithm. The trajectories are colored based on velocity. (a) DDP generates far faster, but more unsafe solutions. (b) MPPI can find more solutions given the higher control variance, but requires more time. (c) \ac{SC-MPPI} can again leverage the underlying safe controller to find higher velocity solutions while maintaining safety. The statistics are shown in \Cref{table: quad large variance}.}
  \label{fig: quadrotor trajectories large variance} 
\end{figure*}

\Cref{fig: quadrotor distance_results_experiment_1} displays the mean and standard deviation of the system from the closest obstacle, as well as to the final target \textit{for all episodes that did not crash}. Clearly, the proposed algorithm \ac{SC-MPPI} demonstrates the lowest final error and a tight distribution to the final target amongst the trials. \ac{MPPI} shows a wide variation in regards to the distance to the closest obstacle, with DDP having a very large variation in the distance to the terminal obstacle for the final state. This large variation stems from the fact that DDP could often diverge away from the arena if it could not find a solution, resulting in the distance from the system to any given obstacle becoming quite large.

\begin{figure} [h]
    \centering
  \subfloat[]{\includegraphics[trim={0.0cm 0.0cm 1.5cm 1.0cm},clip,width=0.5\linewidth]{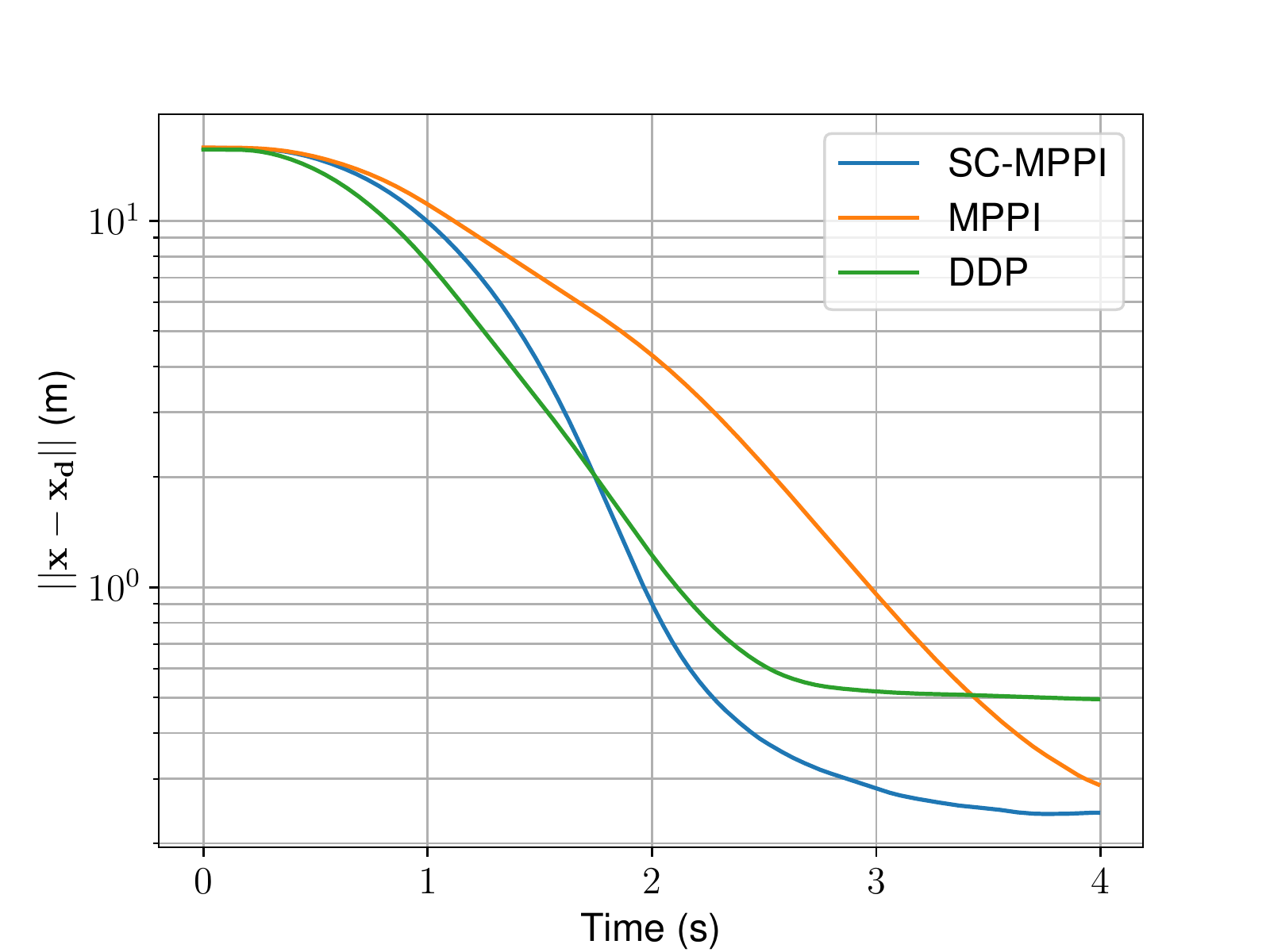}}
  \subfloat[]{\includegraphics[trim={0.0cm 0.0cm 1.5cm 1.0cm},clip,width=0.5\linewidth]{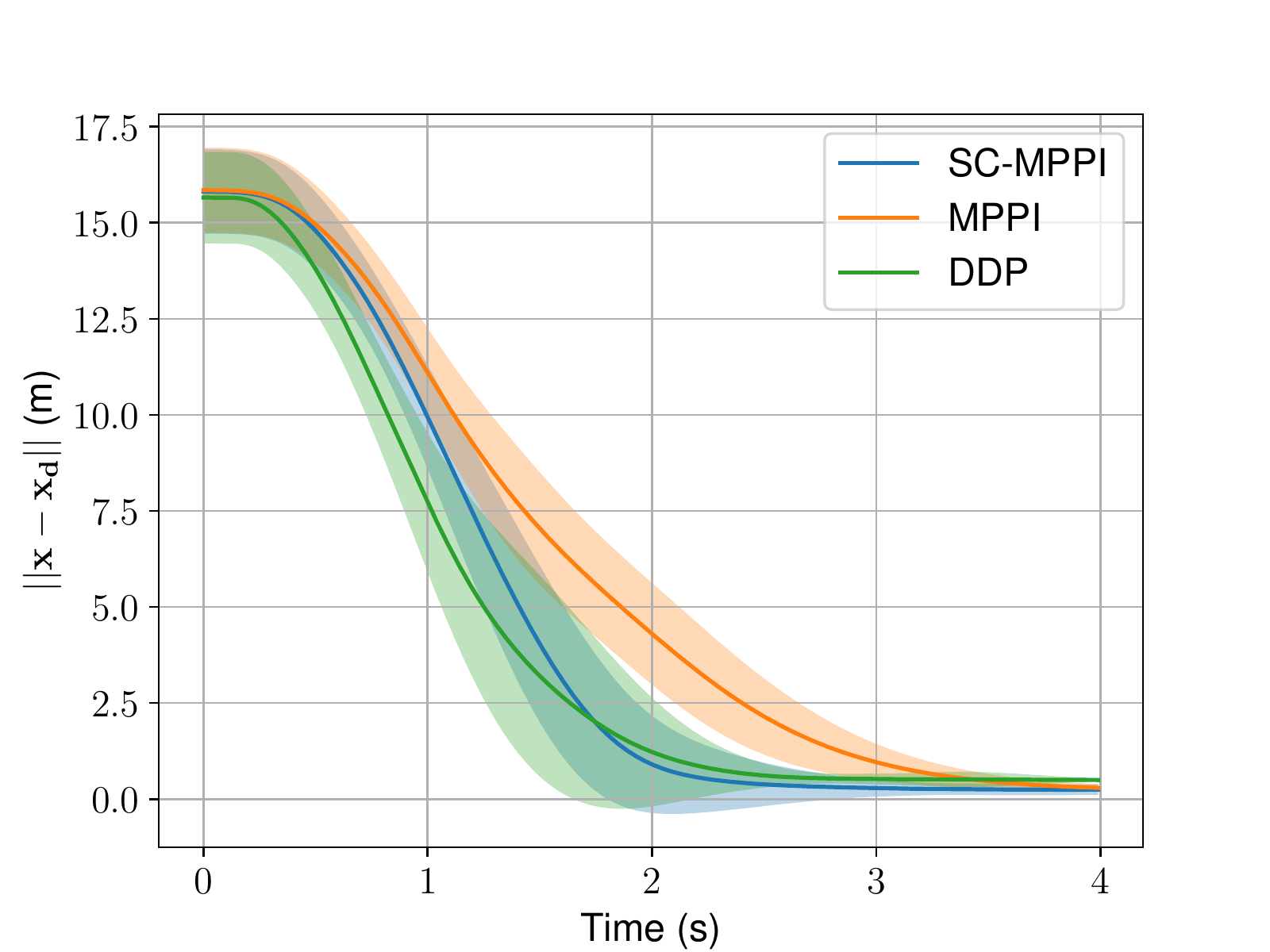}}
  \\
    \subfloat[]{\includegraphics[trim={0.0cm 0.0cm 1.5cm 1.0cm},clip,width=0.5\linewidth]{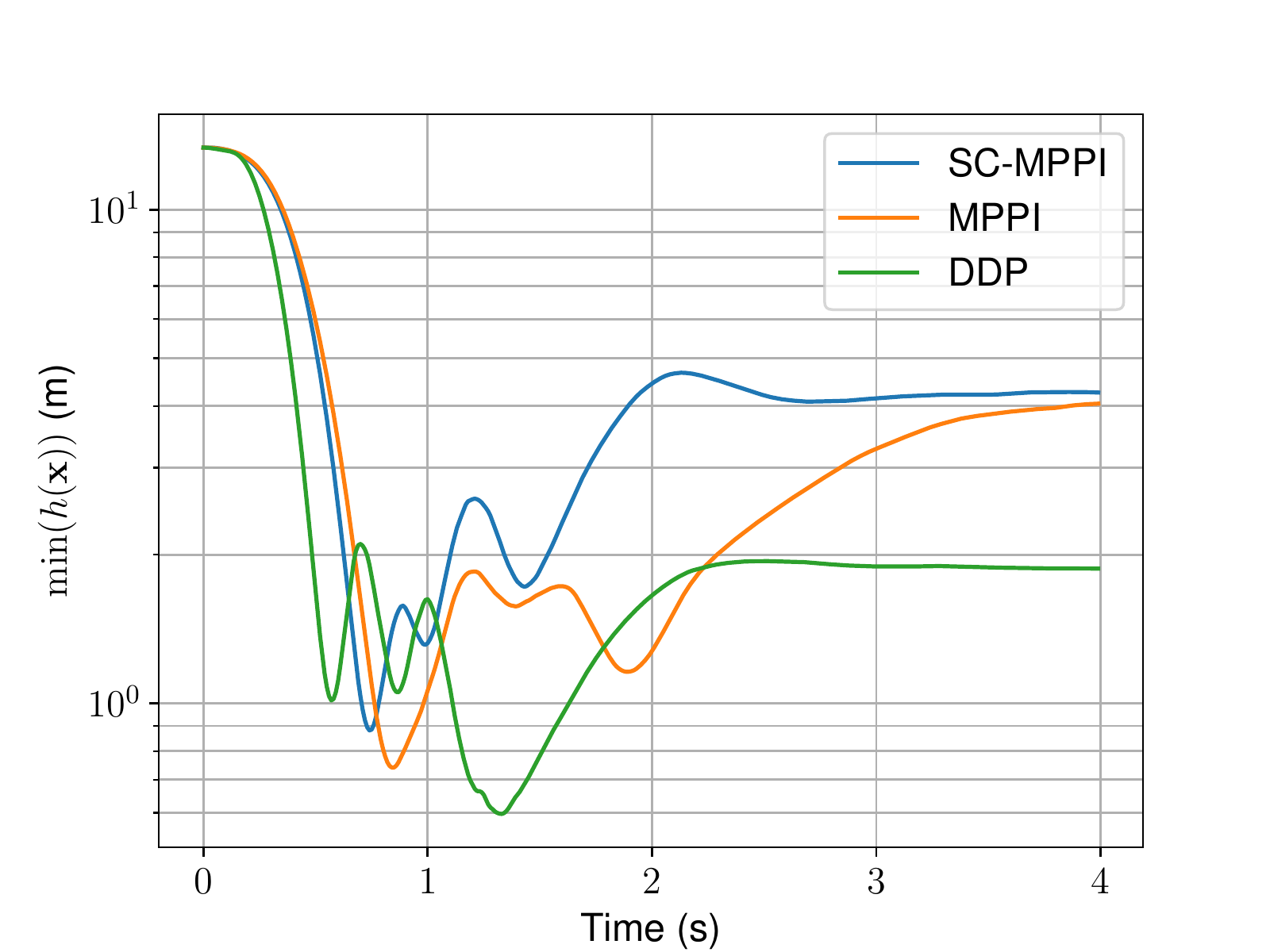}}
  \subfloat[]{\includegraphics[trim={0.0cm 0.0cm 1.5cm 1.0cm},clip,width=0.5\linewidth]{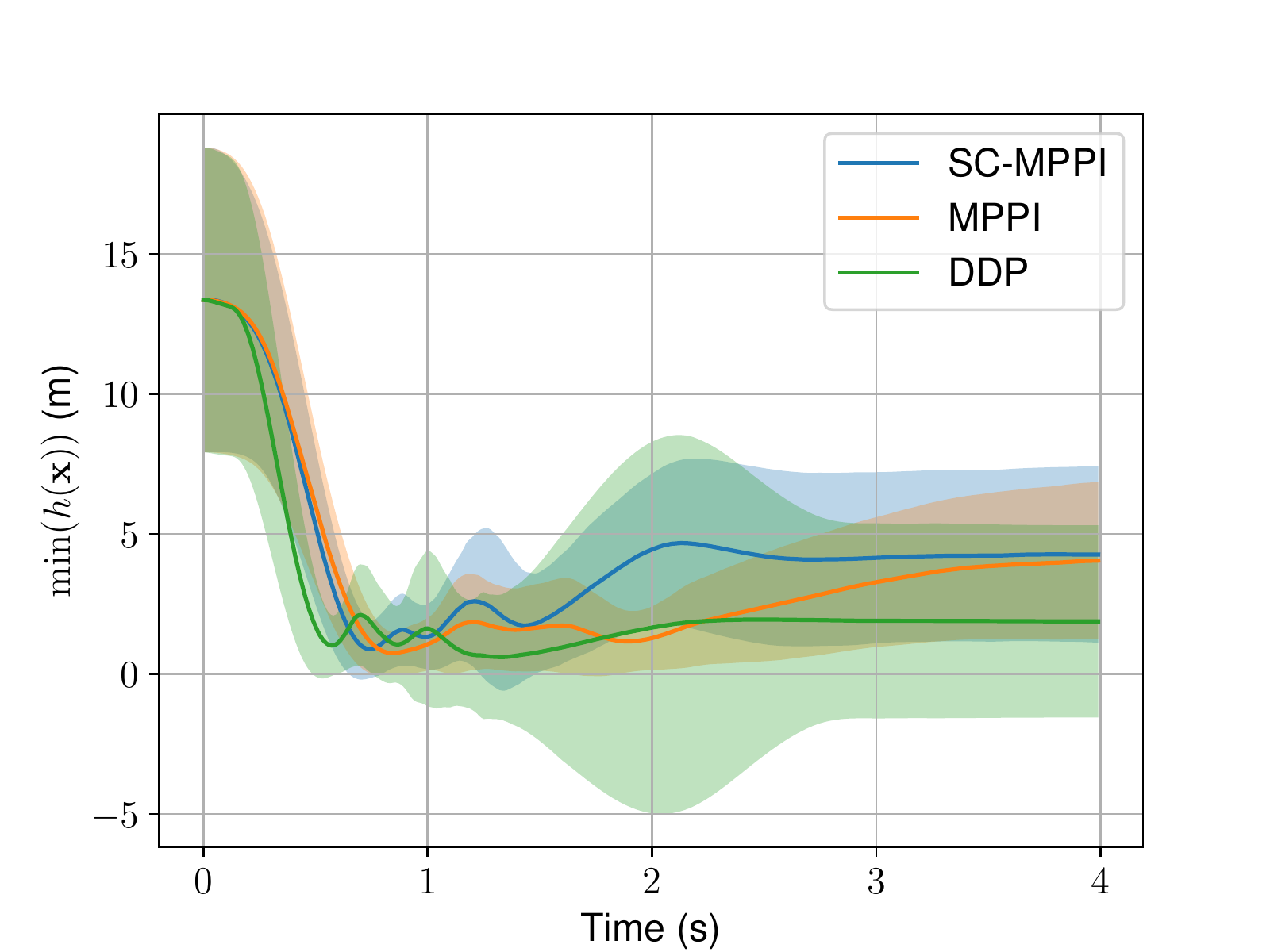}}
  \caption{In the top row is the average distance from the target for the compared algorithms in the logarithmic scale (left) and in the normal scale (right). In the bottom row is the average (over all runs) minimum distance to the obstacles in the logarithmic scale (left) and the normal scale (right). Clearly, \ac{SC-MPPI} achieves more reachability to the target which we hypothesise that its due to the encouraged safe exploration since the target is right behind an obstacle. Additionally, it maintains a further distance from the obstacles compared to MPPI.}
  \label{fig: quadrotor distance results experiment 2} 
\end{figure}

\textbf{Experiment 2:} In this second experiment, the control variance for both \ac{MPPI} and \ac{SC-MPPI} was set to be quite large at $\sigma = [150.0, 150, 50.0, 500.0 ]$, and the time horizon for the overall problem was set to be \textit{4 seconds}. This high control variance turned out to be a necessity for \ac{MPPI} to find a solution. We can see the trajectories of this experiment in \Cref{fig: quadrotor trajectories large variance}. The effect of the increased time horizon and exploration is immediately apparent in the experiment statistics, given in \Cref{table: quad large variance}. Here we can see greatly increased task completion rates for the sampling-based algorithms, and a shift in the performance of safety violation, where \ac{MPPI} has the fewest crashed trajectories. The task completion and final error still favors \ac{SC-MPPI}. Interestingly, the safety sample percentage for MPPI is $45\%$, when averaged across all trajectories and timesteps, whereas for \ac{SC-MPPI} the safe sample percentage is $58.43\%$. Here the correlation between safe sample percentage and performance is less clear. Emipirically we see that sample efficiency of the system matters less than the system having a ``better'' nominal trajectory. Since we are only approximating the true free energy with sampling, a safer initialization appears to lead to more performant solutions. In \Cref{fig: quadrotor distance results experiment 2}, we can see how both DDP and \ac{SC-MPPI} complete the task faster with a tight distribution, but \ac{MPPI} as almost able to achieve the same level of RMS error and distance from the nearest obstacle as \ac{SC-MPPI}. The larger time horizon and larger control variance has a clear benefit for \ac{MPPI} to complete the navigation task.

\begin{table}[h!]
\scriptsize
\caption{\textbf{Experiment 2} Multirotor Dense Navigation Statistical Trial}
\centering
 \begin{tabular}{l |c |c |c} 
  & \ac{DDP} & \ac{MPPI} & \ac{SC-MPPI} \\ 
 \hline\hline
 Compute Time (ms)      & $33.12  \pm (4.98)$     & $\mathbf{3.51 \pm (0.13)}$        & $4.08 \pm (0.11)$ \\ 
  \hline
 Safety Violation \%  & $71.10\%$              & $\mathbf{0.20\%}$        & $3.50\%$\\
  \hline
 Task Completion \%   & $26.10\%$              &  $80.40\%$               & $\mathbf{95.20\%}$ \\
  \hline
  Completion Time (s)   & $2.55 \pm (0.57)$      & $3.39 \pm (0.28)$        & $\mathbf{2.07 \pm (0.33)}$ \\
 \hline
 Position RMSE (m) & $0.49 \pm (0.04)$      & $0.29 \pm (0.10)$        & $\mathbf{0.24 \pm (0.14)}$ \\
 \hline
  Avg Velocity (m/s)    & $\mathbf{5.09 \pm (1.17)}$      & $4.44 \pm (0.31)$        & $4.95 \pm (0.39)$ \\
  \hline
  Max Velocity (m/s)    & $\mathbf{18.17 \pm (2.46)}$     & $11.29 \pm (1.41)$       & $15.38 \pm (1.78)$ \\
 \hline
 \end{tabular}
 \label{table: quad large variance}
\end{table}


\section{Conclusion} \label{sec:conclusion}
In this work, we proposed the idea of importance sampling under safety embedded feedback control. This was then utilized to develop the algorithm \ac{SC-MPPI}. We derived this new algorithm under the principles of information theoretic MPPI and compare it against embedded barrier state MPC-DDP and MPPI. We empirically show that the proposed algorithm can provide a distinct improvement in system performance, system safety, and control exploration, even with lower control variance. Additionally, the algorithm is shown to be computationally feasible to be run in real time, with our experiments demonstrating optimization times of 4-6 milliseconds. \ac{SC-MPPI} does require more optimization time overall when compared to \ac{MPPI}, and the requirement of tuning the additional safety controller can be overkill depending on the problem at hand. We have shown that for difficult, dense navigation tasks, our proposed method can outperform existing techniques. The utilization of a safety controller to improve both the initial importance sampling trajectory for MPPI, as well as maintain the safety of samples moving forward opens the door to further research in safety-critical, sampling-based MPC methods.




\bibliographystyle{plainnat}
\bibliography{references}

\section{Appendix}
\label{Section: Appendix}

\subsection{Discrete Barrier States in Trajectory Optimization for Safe Feedback} \label{App subsec: DBaS in trajectory optimization}
Here we provide details of the barrier state feedback control in optimal control settings.

Consider the optimal control problem 
\begin{align} \begin{split}
    \min_{\utraj}  \sum_{k=0}^{T-1} & \big( q(\vx_k, k) + \vu^{\rT} \Sigma^{-1} \vu \big) + \phi(\vx_T) \\
    \text{subject to } & \vx_{k+1} = F(k, \vx_k, \vu_k ) \\
    & h(\vx_k) > 0
\end{split} \end{align}
Using barrier states as discussed in \autoref{subsec: embedded barrier states}, the problem is transformed to the unconstrained optimal control problem
\begin{align} \begin{split}
    \min_{\utraj}  \sum_{k=0}^{T-1} & \big( q(\bar{\vx}_k, k) + \vu^{\rT} \Sigma^{-1} \vu \big) + \phi(\bar{\vx}_T) \\
    \text{subject to } & \bar{\vx}_{k+1} = \bar{F}(k, \bar{\vx}_k, \vu_k )
\end{split} \end{align}
where $\bar{F} = \begin{bmatrix}F(k,\vx_k, \vu_k), & F^{\beta}\end{bmatrix}^\text{T}$ and $\bar{\vx} = \begin{bmatrix} \vx, & \beta\end{bmatrix}^\text{T}$, as defined in \eqref{eq: safety embedded model}. As mentioned before, one key advantage of augmenting the barrier dynamics is the ability to design a safe feedback controller that is a function of the barrier which is the part of the control law that provides safety (as we will show below). In what follows, we derive DBaS embedded differential dynamic programming (DBaS-DDP) equations in details to highlight the feedback terms in the feedback control equation (which was not derived in details by the authors in \cite{almubarak2022safeddp}). 

First, let us define the embedded model's gradients as
\begin{align}
    \bar{F}_{\bar{\vx}} = \begin{bmatrix} F_{\vx} & 0 \\ F_{\vx}^{\beta} & F_{\beta}^{\beta}  \end{bmatrix}, \qquad \bar{F}_u = \begin{bmatrix} F_{\vu}\\ F_{\vu}^{\beta} \end{bmatrix}
\end{align}
where $F^{\beta}$ is the DBaS dynamics function. Starting with Bellman's equation
\begin{equation} \label{eq: bellemns equation in append}
    V_k(\bar{\vx}_k) = \min_{\vu_k} \{ V_{k+1}(\bar{\vx}_{k+1}) + l_k(\bar{\vx}_k,\vu_k) \}
\end{equation}
Here we define the value function $V(\bar{\vx}) = V(\vx,\beta)$ and its gradient and Hessian as
$$
V_{\bar{\vx}} = \begin{bmatrix}V_{\vx} & V_{\beta} \end{bmatrix}, \qquad V_{\bar{\vx} \bar{\vx}} = \begin{bmatrix}V_{\vx \vx} & V_{\vx \beta} \\ V_{\beta \vx} & V_{\beta \beta} \end{bmatrix}
$$
Expanding the right hand side of Bellman's equation \eqref{eq: bellemns equation in append} two the second order around a nominal trajectory $(\tilde{\vx}, \tilde{\vu})$ yields:
\begin{align*}
    \approx  & (V_{\vx} F_{\vx}+ V_{\beta} F_{\vx}^{\beta} + L_{\vx}) \delta \vx + (V_{\beta} F_{\beta}^{\beta} +L_{\beta}) \delta \beta \\
    & + (V_{\vx} F_{\vu} + V_{\beta} F_{\vu}^{\beta} + L_{\vu}) \delta \vu \\
    & + \frac{1}{2} \delta \vx^{\text{T}} \big( F_{\vx}^{\text{T}} V_{\vx \vx} F_{\vx} + 2 F_{\vx}^{\text{T}} V_{\vx \beta} F_{\vx}^{\beta} + F_\vx^{\beta^{\text{T}}} V_{\beta \beta} F_{\vx}^{\beta} + L_{\vx \vx} \big) \delta \vx \\
    & + \frac{1}{2} \delta \beta^{\text{T}} \big( F_{\beta}^{\beta^{\text{T}}} V_{\beta \beta} F_{\beta}^{\beta} + L_{\beta \beta} \big) \delta \beta \\
    & + \delta \vx^{\text{T}} \big( F_{\vx}^{\text{T}} V_{\vx \beta} F_{\beta}^{\beta} + F_\vx^{\beta^{\text{T}}} V_{\beta \beta} F_{\beta}^{\beta} +L_{\vx \beta} \big) \delta \beta \\
    & + \frac{1}{2} \delta \vu^\text{T} \big( F_{\vu}^{\text{T}} V_{\vx \vx} F_{\vu} + 2 F_{\vu}^{\text{T}} V_{\vx \beta} F_{\vu}^{\beta} + F_{\vu}^{\beta^{\text{T}}} V_{\beta \beta} F_{\vu}^{\beta} + L_{\vu \vu} \big) \delta \vu \\
    & + \delta \vx^\text{T} \big(F_{\vx}^\text{T} V_{\vx \vx} F_{\vu} + F_{\vx}^\text{T} V_{\vx \beta} F_{\vu}^{\beta} + F_\vx^{\beta^{\text{T}}} V_{\beta \vx} F_{\vu} \\
    & \qquad \qquad + F_\vx^{\beta^{\text{T}}} V_{\beta \beta} F_{\vu}^{\beta} + L_{\vx \vu} \big) \delta \vu \\
    & + \delta \beta^\text{T} \big(F_{\beta}^{\beta^{\text{T}}} V_{\beta \vx} F_{\vu} + F_{\beta}^{\beta^{\text{T}}} V_{\beta \beta} F_{\vu}^{\beta} + L_{\beta \vu} \big) \delta \vu
\end{align*}
Define
\begin{align*}
    & Q_{\vx} :=  V_{\vx} F_{\vx}+ {\color{red}{V_{\beta} F_{\vx}^{\beta}}} + L_{\vx} \\
    & Q_{\beta} :=  V_{\beta} F_{\beta}^{\beta} +L_{\beta} \\
    & Q_{\vu} :=  V_{\vx} F_{\vu}+ {\color{red}{V_{\beta} F_{\vu}^{\beta}}} + L_{\vu} \\
    & Q_{\vx \vx} := F_{\vx}^{\text{T}} V_{\vx \vx} F_{\vx} + 2 F_{\vx}^{\text{T}} V_{\vx \beta} F_{\vx}^{\beta} + {\color{red}{F_\vx^{\beta^{\text{T}}} V_{\beta \beta} F_{\vx}^{\beta}}} + L_{\vx \vx} \\
    & Q_{\beta \beta} := F_{\beta}^{\beta^{\text{T}}} V_{\beta \beta} F_{\beta}^{\beta} + L_{\beta \beta}  \\
    & Q_{\vx \beta} := F_{\vx}^{\text{T}} V_{\vx \beta} F_{\beta}^{\beta} + F_\vx^{\beta^{\text{T}}} V_{\beta \beta} F_{\beta}^{\beta} +L_{\vx \beta} \\
    & Q_{\vu\vu} := F_{\vu}^{\text{T}} V_{\vx \vx} F_{\vu} + 2 F_{\vu}^{\text{T}} V_{\vx \beta} F_{\vu}^{\beta} + {\color{red}{F_{\vu}^{\beta^{\text{T}}} V_{\beta \beta} F_{\vu}^{\beta}}} + L_{\vu \vu} \\
    & Q_{\vx\vu} := F_{\vx}^\text{T} V_{\vx \vx} F_{\vu} + F_{\vx}^\text{T} V_{\vx \beta} F_{\vu}^{\beta} + F_\vx^{\beta^{\text{T}}} V_{\beta \vx} F_{\vu} + {\color{red}{F_\vx^{\beta^{\text{T}}} V_{\beta \beta} F_{\vu}^{\beta} }}+ L_{\vx \vu}\\
    & Q_{\beta \vu} := F_{\beta}^{\beta^{\text{T}}} V_{\beta \vx} F_{\vu} + F_{\beta}^{\beta^{\text{T}}} V_{\beta \beta} F_{\vu}^{\beta} + L_{\beta \vu}
\end{align*}

Hence, 
\begin{equation} 
    \frac{\partial V}{\partial \delta \vu} = Q_{\vu} + \delta \vu^\text{T} Q_{\vu\vu} + \delta \vx^\text{T} Q_{\vx\vu} + \delta \beta^\text{T} Q_{\beta \vu}
\end{equation}
and the optimal variation 
\begin{equation}
    \delta \vu^* = - Q_{\vu\vu}^{-1} \Big( Q_{\vu} + Q_{\vu \vx} \delta \vx + {\color{red}{Q_{\vu \beta}}} \delta \beta \Big)
\end{equation}
Substituting the optimal variation back yields the corresponding Riccati equations are
\begin{align}
    \begin{split}
        & V = - \frac{1}{2} Q_{\vu} Q_{\vu\vu}^{-1} Q_{\vu}^\text{T} \\
        & V_{\vx} = Q_{\vx} - Q_{\vu} Q_{\vu\vu}^{-1} Q_{\vu \vx} \\
        & V_{\beta} = Q_{\beta} - Q_{\vu} Q_{\vu\vu}^{-1} Q_{\vu \beta} \\
        & V_{\vx \vx} = Q_{\vx \vx} - Q_{\vx\vu} Q_{\vu\vu}^{-1} Q_{\vu \vx} \\
        & V_{\vx \beta} = Q_{\vx \beta} - Q_{\vx\vu} Q_{\vu\vu}^{-1} Q_{\vu \beta} \\
        & V_{\beta \beta} = Q_{\beta \beta} - Q_{\beta \vu} Q_{\vu\vu}^{-1} Q_{\vu \beta} 
    \end{split}
\end{align}
Note that the DBaS effects both the feed-forward and the feedback gains. 

Similarly, when using other feedback control methods such as the linear quadratic regulator (LQR) for the augmented model, the control law is a function of the barrier which supplies the control signal with safety information. The DBaS-DDP effect is discussed and is shown in the paper when used for importance sampling. To make the barrier state feedback impact on the solution evident and help the reader see it numerically, we illustrate its importance in a numerical example within an LQR problem. The idea of the examples here is to show a detailed numerical example of feedback control laws with barrier states. 

\subsubsection{Quadrotor Optimal Regulation}
Consider the constrained optimal regulation problem:
\begin{align*} \begin{split}
    & \min_{\utraj}  \sum_{k=0}^{T\rightarrow \infty}  \big( 0.5 \vx^{\rT} \vx+ 0.1 \vu^{\rT} \vu \big)  \\
    \text{subject to } & \vx_{k+1} = F(k, \vx_k, \vu_k ) \\
     (x_k - 1)^2 & + (y_k-0.9)^2 + (z_k-1)^2 > 0.4^2 \ \forall k \in [0, T]
\end{split} \end{align*}
where $\vx \in \mathbb{R}^{12}$ is the quadrotor's state vector, $F$ is the discrete dynamics, $x$, $y$ and $z$ are the quadrotor's states position in the three dimensional space respectively. In essence, the problem is to find a control law that would regulate the quadrotor, $\vx \rightarrow \mathbf{0}$ where $\mathbf{0}$ is a vector of zeros of appropriate dimension, for any given initial condition while ensuring that $(x - 1)^2 + (y-0.9)^2 + (z-1)^2 > 0.4^2$ which can represent an obstacle in the three dimensional space. 

Define the safety condition to be $h(\vx_k) = (x_k - 1)^2 + (y_k-0.9)^2 + (z_k-1)^2 - 0.4^2 $, with an inverse barrier $\beta_k = \frac{1}{h(\vx_k)}$. Augmenting the barrier state and dynamics to the model of the system transforms the problem into the unconstrained optimal control problem
\begin{align*} \begin{split}
     \min_{\utraj}  \sum_{k=0}^{T\rightarrow \infty} & \big( 0.5 \vx^{\rT} \vx+ 0.1 \vu^{\rT} \vu + q_\beta \beta_k^2 \big)  \\
    \text{subject to } & \bar{\vx}_{k+1} = \bar{F}(k, \bar{\vx}_k, \vu_k )
\end{split} \end{align*}
Choosing $q_\beta = 10$, linearizing the nonlinear dynamics and solving the new LQR problem, yields the steady state control law $u_k = K_{\infty} \bar{\vx}_k = K_{\vx_{\infty}} \vx_k + {\color{red}{K_{\beta_{\infty}} \beta_k }}$ with 
$$
K_{\beta_{\infty}} = [-1.5514 \quad 0.9765 \quad -0.6901  \quad 2.2293]
$$
The system's states feedback gain matrix is removed due to space limitation. As shown in \autoref{fig: quadrotor lqr}, the barrier states embedded LQR solution (blue trajectory) safely regulate the quadrotor from the starting position (red ball) to the target position (green ball), i.e. avoiding the unsafe region (black sphere), unlike the unconstrained LQR solution (orange trajectory). 

\begin{figure}[h]
    \centering
    \includegraphics[trim={3cm 3cm 3cm 3cm},clip, width=0.7\linewidth]{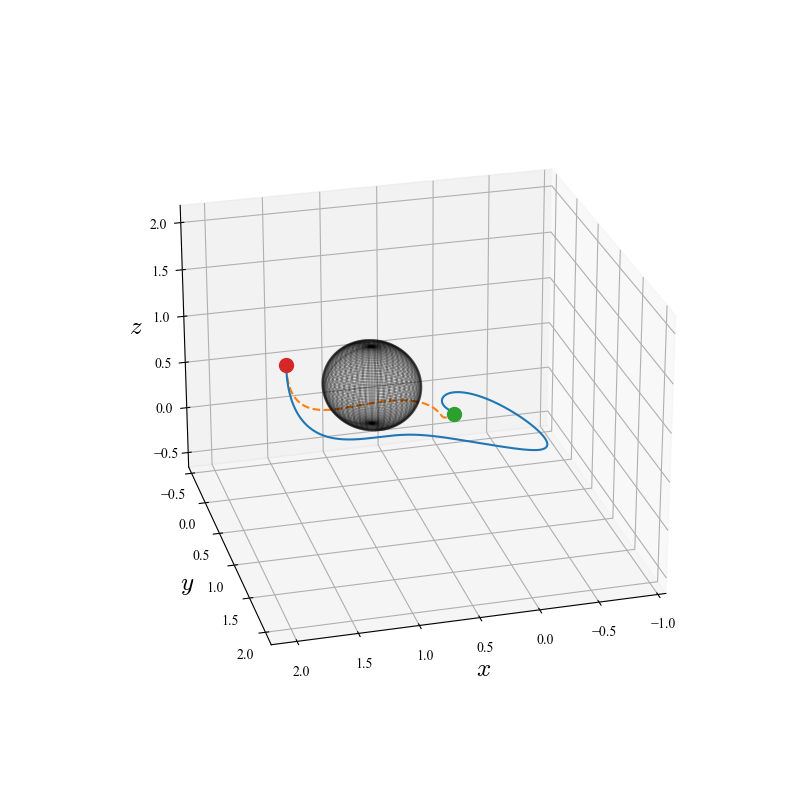} 
    \caption{Quadrotor under LQR control barrier state feedback (blue) and without barrier state feedback (orange), i.e. the unconstrained solution.}
    \label{fig: quadrotor lqr}
\end{figure}

\subsection{Dubins Vehicle Dynamics and Experiment Details} \label{App subsec: Dubins Dynamics}
The Dubins dynamics are given by
\begin{equation*}
    \dot{\vx} = \begin{bmatrix}
        \dot{x} \\ \dot{y} \\ \dot{\theta} 
    \end{bmatrix} = \begin{bmatrix}
        v \cos(\theta) \\ v \sin(\theta) \\ \omega
    \end{bmatrix}
\end{equation*}
where $x$ and $y$ are the vehicle's 2D coordinates, $\theta$ is the yaw (heading) angle, $v$ and $\omega$ are its linear and angular velocities respectively. Each obstacle is represented as an ellipsoid and the constraints to define the barrier states are defined as the distance between the obstacles and the vehicle, e.g. $h_1(\vx) = (x-o_x)^2 + (y-o_y)^2 - r^2 - r_{\text{v}}^2> 0$ where $o_x, \ o_y$ and $r$ are the obstacles center coordinates and radius accordingly and $r_{\text{v}}=0.2$ is the vehicles radius. For computational efficiency purposes, we define a single barrier state for all constraints as explained in \Cref{subsec: embedded barrier states}. The vehicle should start from the initial position $(x_0,y_0) = (0, 0)$ and navigate through the dense field to the target position $(x_f, y_f)= (10,10)$ in $10$ seconds with control limits of $(v_{\text{min}},v_{\text{max}})=(-0.1,10)$ and $(\omega_{\text{min}},\omega_{\text{max}})=(-0.1,10)$, which were set to allow for aggressive driving and not allow it to back up. The sampling time was $dt=0.01$, i.e. a problem horizon of $1000$, and a planning horizon of $50$ time steps, $0.5$ seconds. The cost parameters were selected to be $\lambda =10^{-3}, \alpha=0, \Sigma^{-1} = 300 I_{2 \times 2}, \ Q =\text{diag}(0.2, 0.2, 0.2), \ R= 0.5 \times 10^{-3} I_{2 \times 2}$, $\Phi =\text{diag}(5, 5, 0.1), \ R_{\text{fb}} = 0.5 \times 10^{-2}$ with a maximum number of iterations of $3$. For DBaS-DDP used within SC-MPPI, a quadratic cost was chosen with $Q =\text{diag}(0.1, 0.1, 0), \ q_{\beta} = 10^{-2}, \ R =\text{diag}(5 \times 10^{-3}, 5 \times 10^{-4}), \Phi=\text{diag}(0.02, 0.02, 0)$ with a maximum iterations of $20$.  

\subsection{Multirotor Dynamics and Experiment Details} \label{App subsec: quad control params}
We utilize the dynamics given below, where the inertial positions and velocities are represented by $x, y, z, v_x, v_y, v_z$, respectively, the attitude is represented by the quaternion $\vq = \begin{bmatrix} q_w & q_x & q_y & q_z \end{bmatrix}^\text{T}$, and the body rates of the vehicle are represented by $p, q, r$. Body rate dynamics are simplified in this system, represented via first order system with time constants $\kappa_p, \kappa_q, \kappa_r$. The controls of the system are the desired body rates and system thrust, given in vector $\vu = \begin{bmatrix} p_{\text{des}} & q_{\text{des}} & r_{\text{des}} & \tau \end{bmatrix}$.
\begin{align*}
    \begin{bmatrix} 
    \dot{x} \\ \dot{y} \\ \dot{z}
    \end{bmatrix} &=  
        \begin{bmatrix} 
    v_x \\ v_y \\ v_z
    \end{bmatrix}
\\
        \begin{bmatrix} 
    \dot{v}_x \\ \dot{v}_y \\ \dot{v}_z
    \end{bmatrix} &=
        \frac{1}{m} \mathcal{R}(\vq) \begin{bmatrix} 0 \\ 0 \\ \tau \end{bmatrix} - \begin{bmatrix} 0 \\ 0 \\ g \end{bmatrix}
        \\
        \dot{\vq} &= \frac{1}{2} \begin{bmatrix}
            -q_x & -q_y & -q_z \\ q_w & -q_z & q_y \\ q_z & q_w & -q_x \\ -q_y & q_x & q_w 
        \end{bmatrix} \begin{bmatrix} p \\ q \\ r
        \end{bmatrix}
        \\
            \begin{bmatrix} 
    \dot{p} \\ \dot{q} \\ \dot{r}
    \end{bmatrix} &=  
        \begin{bmatrix} 
    \frac{1}{\kappa_p} (p - p_{\text{des}}) \\ \frac{1}{\kappa_q} (q - p_{\text{des}}) \\ \frac{1}{\kappa_r} (r - r_{\text{des}})
    \end{bmatrix}
\end{align*}

The system's parameters were set to $m=1 \text{kg}$, $g=9.81 \frac{\text{m}}{\text{s}^2}$, $\kappa_p = 0.25$, $\kappa_q = 0.25$, $\kappa_r = 0.7$, and sampling time $dt=0.01$. The multirotor trials were randomized over 950 episodes, where the initial and target positions were varied under uniform box constraints. The vehicle should start from the initial position $(x_0,y_0,z_0) = (-1 \pm 1.5, -1 \pm 1.5, -4 \pm 0.1)$ and navigate through the dense field to the target position $(x_f, y_f, z_f)= (11 \pm 1.5, 11 \pm 1.5, 3 \pm 0.3)$ in $3$ seconds with control limits of $(p_{\text{min}},p_{\text{max}})=(-10,10) \frac{\text{rad}}{s}$,  $(q_{\text{min}},q_{\text{max}})=(-10,10)\frac{\text{rad}}{s}$,  $(r_{\text{min}},r_{\text{max}})=(-10,10)\frac{\text{rad}}{s}$,  $(\tau_{\text{min}},\tau_{\text{max}})=(0,45) N$. The sampling time was $dt=0.01$, i.e. a problem horizon of $300$, and a planning horizon of $75$ time steps, $0.75$ seconds.

\subsubsection{Experiment 1 Implementation Details}
For this experiment, we ran MPPI and SC-MPPI using the same parameters but with the difference being that MPPI has a barrier cost while SC-MPPI does not but instead its importance sampling is embedded with a safe feedback control effectively decoupling safety from MPPI performance optimization. The parameters chosen were as follows:
\begin{align*}
    \begin{split}
        & \lambda = 0.01, \ \alpha =0.7 \\
        & \Sigma^{-1} = \text{diag} ( 5, 5, 5, 15) \\ 
        & Q  = \text{diag} (1.50, 1.50, 1.70,10.5,  10.5,  10.5,1, 0, 0, 0, 1, 1, 1 ) \\ 
        & q_\beta = 10^{-6} \\
        & R  = \text{diag} (50 , 50 , 500 , 5 ) \\ 
        & \Phi  = \text{diag} (2500, 2500, 3000,9.5,  9.5,  1.000,0, 10, 10, 10,10, 10, 10) \\ 
        & R_{\text{bf}} = \text{diag} (20, 20, 20, 200) 
    \end{split}
\end{align*}
with a maximum number of iteration of $1$. For DBaS-DDP within SC-MPPI for the barrier state feedback, it had a maximum number of iterations of $5$ and the following cost parameters: 
\begin{align*}
    \begin{split}
        & Q  = \text{diag} (10, 10, 20, 10^{-2},  10^{-2},  10^{-1}, 0, 0, 0, 0,10^{-3}, 10^{-3}, 10^{-3}) \\
        & q_{\beta} = 10^{-4} \\
        & R  = \text{diag} (10, 10, 200.0, 1 ) \\ 
        & \Phi  = \text{diag} (10^{2}, 10^{2}, 10^{2},
                      1,  1,  1,
                      0, 0, 0, 0,
                      10^{-2}, 10^{-2}, 10^{-2}) \\ 
    \end{split}
\end{align*}
For \ac{DBaS} embedded \ac{MPC-DDP}, the following parameters were chosen, 
\begin{align*}
    \begin{split}
        & Q  = \text{diag} (2.5, 2.5, 50,
                           15,  15,  25,
                           1, 0, 0, 0,
                           3, 3, 3) \\
        & q_{\beta} = 15 \\
        & R  = \text{diag} (550, 550, 5500, 900) \\ 
        & \Phi  = \text{diag} (350, 350, 450,
                      15,  15,  100,
                      100, 0, 0, 0,
                      30, 30, 30) \\ 
    \end{split}
\end{align*}
with a maximum number of iteration of $25$.

\subsubsection{Experiment 2 Implementation Details}
MPPI and SC-MPPI parameters chosen were as follows:
\begin{align*}
    \begin{split}
        & \lambda = 0.01, \ \alpha_{\text{MPPI}} =0.7, \  \alpha_{\text{SC-MPPI}} =0.9 \\
        & \Sigma^{-1} = \text{diag} ( 150, 150, 50, 500) \\ 
        & Q_{\text{MPPI}}  = \text{diag} (1.50, 1.50, 1.70,10.5,  10.5,  10.5,1, 0, 0, 0, 1, 1, 1 ) \\ 
        & Q_{\text{SC-MPPI}}  = \text{diag} (150, 150, 170,10.5,  10.5,  10.5,1, 0, 0, 0, 1, 1, 1 ) \\ 
        & q_\beta = 10^{-4} \\
        & R  = \text{diag} (500 , 500 , 5000 , 50 ) \\ 
        & \Phi  = \text{diag} (2500, 2500, 3000,9.5,  9.5,  1.000,0, 10, 10, 10,10, 10, 10) \\ 
        & R_{\text{bf}} = \text{diag} (20, 20, 20, 200) 
    \end{split}
\end{align*}
with a maximum number of iteration of $1$. For DBaS-DDP within SC-MPPI for the barrier state feedback, it had a maximum number of iterations of $5$ and the following cost parameters: 
\begin{align*}
    \begin{split}
        & Q  = \text{diag} (10, 10, 20, 10^{-2},  10^{-2},  10^{-1}, 0, 0, 0, 0,10^{-3}, 10^{-3}, 10^{-3}) \\
        & q_{\beta} = 10^{-4} \\
        & R  = \text{diag} (10, 10, 200.0, 1 ) \\ 
        & \Phi  = \text{diag} (10^{2}, 10^{2}, 10^{2},
                      1,  1,  1,
                      0, 0, 0, 0,
                      10^{-2}, 10^{-2}, 10^{-2}) \\ 
    \end{split}
\end{align*}
For \ac{DBaS} embedded \ac{MPC-DDP}, the following parameters were chosen, 
\begin{align*}
    \begin{split}
        & Q  = \text{diag} (2.5, 2.5, 50,
                           1.5,  1.5,  2.5,
                           1, 0, 0, 0,
                           3, 3, 3) \\
        & q_{\beta} = 15 \\
        & R  = \text{diag} (550, 550, 5500, 900) \\ 
        & \Phi  = \text{diag} (3500, 3500, 4500,
                      15,  15,  100,
                      100, 0, 0, 0,
                      30, 30, 30) \\ 
    \end{split}
\end{align*}
with a maximum number of iteration of $25$.



\end{document}

%% file: packages.tex
\usepackage{amsmath}
\usepackage{amssymb}
\usepackage[utf8]{inputenc}
\usepackage[T1]{fontenc}
\usepackage{graphicx}
\usepackage{enumerate}
\usepackage{xcolor}
\usepackage{comment}
\usepackage{multicol}
\usepackage{hyperref}
\usepackage{tabularx}
\usepackage{makecell}
\usepackage{colortbl}
\usepackage{tablefootnote}
\usepackage{multirow}
\usepackage{paralist}
\usepackage{placeins}
\usepackage{nameref}
\usepackage{cancel}
\usepackage{subcaption}
\usepackage{multicol}
\usepackage[numbers]{natbib}
\usepackage[ruled,vlined]{algorithm2e}
\usepackage[font=small]{caption}

\usepackage{float} 

\usepackage{cleveref}

\usepackage[nolist]{acronym}

\newtheorem{lemma}{Lemma}
\newtheorem{theorem}{Theorem}[section]
\newtheorem{definition}{Definition}[section]
\newtheorem{proposition}[theorem]{Proposition}
\hypersetup{
    colorlinks=true,
    linkcolor=blue,
    filecolor=magenta,      
    urlcolor=cyan,
}


%% file: commands.tex
\newcommand{\vx}[0]{\mathbf{x}}

\newcommand{\vu}[0]{\mathbf{u}}

\newcommand{\vq}[0]{\mathbf{q}}
\newcommand{\vv}[0]{\mathbf{v}}

\newcommand{\argmin}{\operatornamewithlimits{argmin}}

\newcommand{\rT}[0]{{\textrm{T}}}

\newcommand{\xtraj}[0]{\mathbf{x}_{\tau}}

\newcommand{\utraj}[0]{\mathbf{u}_{\tau}}

\newcommand{\real}[0]{\mathbb{R}}
\newcommand{\domain}[0]{\mathcal{D}}
\newcommand{\safeset}[0]{\mathcal{C}}

\newcommand{\Pb}[0]{\mathbb{P}}
\newcommand{\Qb}[0]{\mathbb{Q}}

\newcommand{\KL}[2]{\mathbb{KL}\left({#1}\parallel {#2}\right)}

\newcommand{\ExP}[2]{\mathbb{E}_{#1}\big[ #2 \big]}

\newcommand{\pluseq}{\mathrel{+}=}
\newcommand{\rd}[0]{\text{d}}

\begin{acronym}
\acro{MPPI}{Model-Predictive Path Integral Control}
\acro{Tube-MPPI}{Tube Model-Predictive Path Integral Control}
\acro{RMPPI}{Robust Model-Predictive Path Integral Control}
\acro{CCM}{Control Contraction Metrics}
\acro{MPC}{Model-Predictive Control}
\acro{IS}{Importance Sampling}
\acro{KL}{Kullback–Leibler}
\acro{GT-ARF}{Georgia Tech Autonomous Racing Facility}
\acro{iLQG}{iterative Linear Quadratic Gaussian}
\acro{AIS}{Augmented Importance Sampling}
\acro{SC-IS}{Safety-Controlled Importance Sampling}
\acro{CBFs}{Control Barrier Functions}
\acro{SC-MPPI}{Safety Controlled Model Predictive Path Integral Control}
\acro{BAS}{BArrier State}
\acro{DDP}{Differential Dynamic Programming}
\acro{MPC-DDP}{ Model Predictive Differential Dynamic Programming}
\acro{RMSE}{Root Mean-Squared Error}
\acro{DBaS}{Discrete Barrier States}
\end{acronym}
